\documentclass[11pt]{article}
\usepackage{fullpage}
\usepackage[shortlabels]{enumitem}
\usepackage{titling}
\usepackage{stmaryrd}
\usepackage{palatino} 
\usepackage{mathpazo}
\usepackage{braket}
\usepackage{amsfonts}
\usepackage{amssymb}
\usepackage{amsmath}
\usepackage{mathrsfs} 
\usepackage{mathabx}
\usepackage{bbm}
\usepackage{latexsym}
\usepackage{amsthm}
\usepackage{complexity}
\usepackage{algorithm}
\usepackage[noend]{algpseudocode}
\usepackage[usenames]{color}
\usepackage{hyperref, etoolbox}
\usepackage{subcaption}
\usepackage{verbatim}
\usepackage{ragged2e, graphicx, pifont}
\usepackage{indentfirst}
\usepackage{thm-restate}	
\usepackage{caption}
\usepackage{geometry}
\hypersetup{
colorlinks = true,
citecolor= blue
}

\usepackage{tikz-cd}
\usetikzlibrary{decorations.pathmorphing}
\usetikzlibrary{arrows.meta}

\usepackage[framemethod=TikZ]{mdframed}
\usepackage{extarrows}
\usepackage{arydshln}

\newcommand{\eps}{\varepsilon}

\newcommand{\defeq}{{\text{}:=}}

\newcommand{\OWF}{\mathrm{OWF}}
\newcommand{\EFI}{\mathrm{EFI}}
\newcommand{\OWSG}{\mathrm{OWSG}}
\newcommand{\sOWSG}{\textnormal{sv-OWSG}}
\newcommand{\PRSG}{\mathrm{PRSG}}
\newcommand{\PRG}{\mathrm{PRG}}
\newcommand{\support}{\mathrm{support}}
\newcommand{\Id}{\mathbb{1}}

\renewcommand{\E}{\mathbb{E}}

\newcommand{\tab}{\hspace{2cm}}
\newcommand{\supp}{\mathrm{supp}}

\newcommand{\Tr}{\mathrm{Tr}}

\renewcommand{\poly}{\mathrm{poly}}

\newcommand{\ketbra}[2]{\ket{#1}\!\!\bra{#2}}

\newcommand{\good}{\ensuremath{\mathsf{Good}}}

\newtheorem{theorem}{Theorem}

\newtheorem{cor}[theorem]{Corollary}

\newtheorem{fact}{Fact}
\newtheorem{claim}{Claim}

\theoremstyle{definition}
\newtheorem{definition}{Definition}

\newcommand{\negl}{\mathsf{negl}}

\newcommand{\QPT}{\ensuremath{\mathrm{QPT}}}

\newcommand {\diver} [2] {{\mathrm{D}}(#1 \| #2)}

\newcommand {\dmax} [2] {{\mathrm{D}_{\infty}}(#1 \| #2)}
\newcommand {\dtwo} [2] {{\mathrm{D}_{2}}(#1 \| #2)}

\newcommand{\bern}{\textsf{Bern}}

\newcommand{\onenorm}[2]{\left\Vert #1 - #2\right\Vert_1}

\newcommand{\ext}{\mathsf{Ext}}
\newcommand{\gen}{\textsc{StateGen}}
\newcommand{\keygen}{\textsc{KeyGen}}
\newcommand{\ver}{\textsc{Ver}}
\usepackage{graphicx,stackengine,scalerel}
\usetikzlibrary{shapes.misc, positioning}
\def\tang{\ThisStyle{\abovebaseline[0pt]{\scalebox{-1}{$\SavedStyle\perp$}}}}

\title{Commitments are equivalent to statistically-verifiable one-way state generators}

\author{Rishabh Batra$^1$  \and Rahul Jain$^{1,2, 3}$  }

\date{${}^{1}$ Centre for Quantum Technologies\\
${}^{2}$Department of Computer Science, National University of Singapore\\
${}^{3}$ MajuLab, UMI 3654, Singapore}

\begin{document}
\maketitle

\begin{center}
ArXiv: 2404.03220 \quad ; \quad Appeared at FOCS 2024
\end{center}
\begin{abstract}

One-way state generators ($\OWSG$) \cite{owsg} are natural quantum analogs to classical one-way functions. {We consider statistically-verifiable $\OWSG$s ($\sOWSG$), which are potentially weaker objects than $\OWSG$s.} We show that $O\left(\frac{n}{\log(n)}\right)$-copy $\sOWSG$s ($n$ represents the input length) are equivalent to $\poly(n)$-copy $\sOWSG$s and to quantum commitments. Since known results show that  $o\left(\frac{n}{\log(n)}\right)$-copy $\OWSG$s cannot imply commitments~\cite{cavalar2023computational}, this shows that $O\left(\frac{n}{\log(n)}\right)$-copy $\sOWSG$s are the weakest $\OWSG$s from which we can get commitments (and hence much of quantum cryptography).

Our construction follows along the lines of H\r{a}stad, Impagliazzo,  Levin and  Luby~\cite{HILL}, who obtained classical pseudorandom generators ($\PRG$) from classical one-way functions ($\OWF$), however with crucial modifications. Our construction, when applied to the classical case, provides an alternative to the construction provided by~\cite{HILL} to obtain a classical mildly non-uniform $\PRG$ from any classical $\OWF$. Since we do not argue conditioned on the output $f(x)$, our construction and analysis is arguably simpler and may be of independent interest. For converting a mildly non-uniform $\PRG$ to a uniform $\PRG$, we can use the same construction as~\cite{HILL}.

\end{abstract}
\section{Introduction}
Commitments are one of the most fundamental primitives in cryptography. A commitment scheme allows a sender to commit some message $m$ (in an encrypted form) to a receiver such that the following properties hold:

\textbf{Hiding}: the receiver should not be able to learn the message $m$ before the sender decides to reveal it. 

\textbf{Binding}: the sender should not be able to change the message to some $m'\neq m$ once it has committed to $m$.

An intuitive way to think about commitments is as follows: the sender Alice locks the message to be sent in a physical box and sends it across to the receiver Bob and keeps the key with herself. 
When Alice wants to reveal the message, she just sends the key across to Bob. Hiding follows from the fact that Bob cannot open the box unless Alice sends the key. Binding follows from the fact that Alice can't change the message once she has sent the box across.

There are two types of securities considered in cryptography: statistical security (against computationally unbounded adversaries) and computational security (against computationally bounded adversaries). Most of the results in classical cryptography are based on some computational assumptions.

It is known that both statistical hiding and statistical binding cannot be achieved using classical protocols. There are some tasks e.g. statistically secure key distribution that are not possible using classical protocols, however are possible using quantum protocols~\cite{BB84,Ekert91}. This leads to the question if it is possible to have both statistical hiding and statistical binding using quantum protocols for commitments. It is known that
this is also not possible~\cite{Mayers_BC_Impossible, Lo_imposibility_stuff, LO_Chahu_BC_Impossible}. Thus in the quantum regime also, computational assumptions are necessary for the existence of many cryptographic primitives. 

A pseudorandom generator ($\PRG$) is an efficiently computable function, $\PRG: \{0,1\}^n \rightarrow \{0,1\}^{l(n)}$ such that $l(n)>n$ and $\PRG(U_n)$ ($U_n$ represents the uniform distribution on $n$ bits) is computationally indistinguishable from $U_{l(n)}$. A one-way function ($\OWF$) is a function that is efficiently computable and hard to invert. 
The existence of one-way functions is a widely used computational assumption in classical cryptography and many of the classical cryptographic primitives are known to imply the existence of one-way functions. An EFID pair~\cite{EFID} is a pair of efficiently samplable distributions that are statistically far and computationally indistinguishable. It is known that one-way functions, EFID and $\PRG$ are equivalent to each other~\cite{Goldreich_2001}. 

In the quantum world, however, the story is not as straightforward. Morimae and Yamakawa \cite{Morimae_2022} showed that quantum commitments can be based on pseudorandom state generators ($\PRSG$) \cite{JLS_oneway, OT3}, which can be thought of as a generalization of $\PRG$s with quantum outputs. It is known that $\PRSG$s can exist even if BQP = QMA (relative to a quantum oracle)~\cite{Kre} or if P = NP (relative to a classical oracle)~\cite{KSTQ}, which does not allow for the existence of one-way functions (relative to these oracles).

An $\EFI$ (see Definition~\ref{def:imEFI}) is a quantum polynomial time ($\QPT$) algorithm that generates a pair of quantum states that is statistically far and computationally indistinguishable~\cite{brakerski2022computational}. Quantum commitments are known to be equivalent to EFI~\cite{brakerski2022computational, EFI}. A  $k^*$-imbalanced $\EFI$ is a $\QPT$ algorithm that takes as input a parameter $k$ and generates a pair of states, say $\rho_0(k)$ and $\rho_1(k)$, that is computationally indistinguishable for $k\leq k^*$ and is statistically far for $k\geq k^*$. It is easily seen that an $\EFI$ implies an imbalanced-$\EFI$. Khurana and Tomer~\cite{KT} show that the converse is also true (below $\lambda$ represents the security parameter).  
\begin{fact}[\cite{KT}]\label{fact:imbalanced_EFI_to_EFI-1}
Let $k^*(\lambda) \in \poly(\lambda)$. Then a $k^*$-imbalanced $\EFI$ implies $\EFI$. 
\end{fact}
One-way state generators ($\OWSG$) \cite{owsg} are quantum generalizations of one-way functions in which the output $\phi(x)$ (for input $x\in \{0,1\}^n$) is a quantum state and the function is difficult to invert even when multiple copies of the output state, that is $\phi(x)^{\otimes m}$, is given. Note that for classical one-way functions, providing many copies of $f(x)$ to the adversary was not changing anything since the adversary could have made the copies itself. However, the quantum no-cloning theorem does not allow the adversary to copy quantum states and hence providing many copies may help the adversary. 

We use a definition of statistically-verifiable $\OWSG$ ($\sOWSG$, see Definition~\ref{def:OWSG}) that differs from the definitions used in~\cite{owsg} and~\cite{KT}. In our definition, the verifier $\ver$ is an unbounded algorithm and the adversary $A$ is a non-uniform $\QPT$ algorithm, that is a $\QPT$ algorithm that is provided a polynomial-size classical advice string. Providing polynomial-size advice string to adversaries is quite common in classical cryptography~\cite{Goldreich_2001} and allows for several reductions to go through. In the definition by~\cite{owsg}, both $\ver$ and $A$ are uniform $\QPT$ algorithms. In the definition by~\cite{KT}, $\ver$ is uniform $\QPT$ and $A$ is a $\QPT$ algorithm that is provided polynomial-size quantum states as advice. {Note that increasing the power of $\ver$ allows for a larger class of $\OWSG$s and increasing the power of $A$ allows for a smaller class of $\OWSG$s.} In that sense, our definition is incomparable to that of~\cite{owsg}. { An $\OWSG$ according to~\cite{KT} is also an $\sOWSG$.}

\cite{KT} show that $\OWSG$s with pure-state output imply $\EFI$. Construction of an $\EFI$ from a mixed-state output $\OWSG$ was left open. {We resolve this question by constructing EFI from $\sOWSG$ (which is implied by $\OWSG$). } Our main result is (below $n$ represents the key-length of the $\sOWSG$):
\begin{theorem}\label{thm:main}
   Let $m = \frac{cn}{\log(n)}$ for some constant $c>0$. An $m$-copy {$\sOWSG$} implies a $k^*$-imbalanced $\EFI$  for some $k^*(\lambda) \in \poly(\lambda)$. 
\end{theorem}
From Theorem \ref{thm:main} and Fact \ref{fact:imbalanced_EFI_to_EFI-1}, we get:
\begin{cor}\label{cor:forward}
        Let $m = \frac{cn}{\log(n)}$ for some constant $c>0$. An $m$-copy $\sOWSG$ implies an $\EFI$.
\end{cor}
 Combining this with previous works that show  {oblivious transfer} and secure multi-party computation  from quantum commitments~\cite{OT1, OT2, OT3,brakerski2022computational} (which are in turn equivalent to EFI), we get: 
 \begin{cor}
       Let $m = \frac{cn}{\log(n)}$ for some constant $c>0$. An $m$-copy $\sOWSG$ implies {oblivious transfer} and secure multi-party computation for all functionalities.
 \end{cor}
We show the converse to Corollary \ref{cor:forward} is also true.
\begin{theorem}
    An $\EFI$  implies a $\poly(n)$-copy $\sOWSG$.
\end{theorem}
Combining the above results we get: 
\begin{cor}
    The following cryptographic primitives are equivalent:
    \begin{itemize}
        \item $\poly(n)$-copy $\sOWSG$,
         \item $O\left(\frac{n}{\log(n)}\right)$-copy $\sOWSG$,
         \item $\EFI$,
         \item quantum commitments.
    \end{itemize}
\end{cor}
A natural question that can be asked at this point is whether an $o\left(\frac{n}{\log(n)}\right)$-copy $\sOWSG$ implies a $\poly(n)$-copy $\sOWSG$ (equivalently $\EFI$)?  Cavalar et al.~\cite{cavalar2023computational} show that $o\left(\frac{n}{\log(n)}\right)$-copy pure-state output $\OWSG$s ({that are potentially stronger than $\sOWSG$s}) with statistical security exist unconditionally. Hence they cannot imply an $\EFI$  unless $\EFI$s also exist unconditionally. This gives the following corollary.
 \begin{cor}
Unless $\EFI$s exist unconditionally, $O\left(\frac{n}{\log(n)}\right)$-copy $\sOWSG$s are the weakest $\OWSG$s from which we can get an $\EFI$, equivalently, commitment schemes. 
\end{cor}
Figures \ref{fig:before} and \ref{fig:after} capture relations between different primitives before and after our work. The black arrow from primitive $A$ to $B$ represents that if primitive $A$ exists, then primitive $B$ also exists. The red dotted lines represent separations, either oracle or unconditional (unless $\EFI$s exist unconditionally), between the primitives.
\begin{figure}
    \centering

\resizebox{16cm}{5cm}{\begin{tikzpicture}[font=\small, thick]
\node[draw,  rectangle, rounded corners, minimum width=2cm,minimum height=1cm]  at (1.5,0) (XX)  { OWF };
\node[draw,  rectangle, rounded corners , minimum width=2cm,minimum height=1cm] at (4.5,0) (YY) { pure-OWSG  };
\node[draw,  rectangle, rounded corners , minimum width=2cm,minimum height=1cm] at (6,-2) (ZZ) { EFI   };
\node[draw,  rectangle, rounded corners , minimum width=2cm,minimum height=1cm] at (6,2) (AA) { {OWSG} };
\node[draw,  rectangle, rounded corners , minimum width=2cm,minimum height=1cm] at (10,2) (A) {  sv-OWSG };
\node[draw,  rectangle, rounded corners , minimum width=2cm,minimum height=1cm] at (10,-2) (C) { Commitments };
\node[draw,  rectangle, rounded corners , minimum width=2cm,minimum height=1cm] at (14.5,1) (C1) { Oblivious Transfer/MPC };
\draw [ultra thick, ->](XX)--(YY);
\draw [ultra thick, ->](YY) -- (ZZ);
\draw [red, ultra thick, dashed, ->](YY.165) -- (XX.15);
\draw [ultra thick, ->](YY) -- (AA);
\draw [ultra thick, ->](AA) -- (A);
\draw [ultra thick, ->](ZZ) -- (C);
\draw [ultra thick, ->](C) -- (ZZ);
\draw [ultra thick, ->](C1) -- (C);
\draw [ultra thick, ->](C) -- (C1);
\node at (3,0.8) {\cite{Kre,KSTQ}, oracle};
  \node at (4.5,-1) {\cite{KT}};
  \node at (7.9,-1.2) {\cite{brakerski2022computational, EFI}};
    \node at (15,-0.8) {\cite{brakerski2022computational, OT1, OT2, OT3}};
\end{tikzpicture}}
 \caption{Relations between different primitives before our work }
 \label{fig:before}
\end{figure}

\begin{figure}
    \centering

\resizebox{16cm}{5cm}{
\begin{tikzpicture}
\node[draw,  rectangle, rounded corners , minimum width=2cm,minimum height=1cm]  at (0,0) (XX)  { OWF };
\node[draw,  rectangle, rounded corners , minimum width=2cm,minimum height=1cm] at (3,0) (YY) { pure-OWSG };
\node[draw,  rectangle, rounded corners , minimum width=2cm,minimum height=1cm] at (1,-4) (Y) { sv-pure-OWSG };
\node[draw,  rectangle, rounded corners , minimum width=2cm,minimum height=1cm] at (6,0) (ZZ) { OWSG };
\node[draw,  rectangle, rounded corners , minimum width=2cm,minimum height=1cm] at (9,0) (AA) { {sv-OWSG} };
\node[draw,  rectangle, rounded corners , minimum width=3cm,minimum height=1cm] at (12.7,0) (A) {  $O(\frac{n}{\log(n)})$-sv-OWSG };
\node[draw,  rectangle, rounded corners , minimum width=2cm,minimum height=1cm] at (16.5,0) (B) { EFI };
\node[draw,  rectangle, rounded corners , minimum width=2.5cm,minimum height=1cm] at (15,-4) (C) { $o(\frac{n}{\log(n)})$-sv-OWSG};
\node[draw,  rectangle, rounded corners , minimum width=2.7cm,minimum height=1cm] at (6,-4) (E) { $o(\frac{n}{\log(n)})$-pure-OWSG};
\node[draw,  rectangle, rounded corners , minimum width=2cm,minimum height=1cm] at (10,-4) (D) { $o(\frac{n}{\log(n)})$-OWSG};
\draw [ultra thick, scale=10, {->}](XX)--(YY);
\draw [ultra thick, ->](Y)--(YY);
\draw [ultra thick, ->](YY)--(Y);
\draw [ultra thick, ->](YY) -- (ZZ);
\draw [red, ultra thick, dashed, ->](YY.165) -- (XX.15);
\draw [red, ultra thick, dashed, ->](E) -- (B);
\draw [ultra thick, ->](YY) -- (ZZ);
\draw [ultra thick, ->](ZZ) -- (AA);
\draw [ultra thick, ->](D) -- (C);
\draw [ultra thick, ->](AA) -- (A);
\draw [ultra thick, ->](A) -- (B);
\draw [ultra thick, ->](A) -- (AA);
\draw [ultra thick, ->](B) -- (A);
\draw [ultra thick, ->](YY) -- (E);
\draw [ultra thick, ->](E) -- (D);
\draw [ultra thick, ->](A) -- (C);
\node at (10.5,1) {Our work};
\node at (15,1) {Our work};
\node at (1.5,1) {\cite{Kre,KSTQ}, oracle};
   \node at (9,-1.7) {\cite{cavalar2023computational}, unconditional };
\end{tikzpicture}}
\caption{Relations between different primitives after our work }\label{fig:after}
\end{figure}

\subsection*{Proof idea:}

Let $f(\cdot)$ be an $\OWF$ that is hard to invert given $f(X)$ ($X$ represents the input random variable). A hardcore function $g(\cdot)$ for $f(\cdot)$ is an efficiently computable function such that $f(X)g(X)$ and $f(X)\otimes U_{|g(X)|}$ are computationally indistinguishable. However, we do not have the guarantee that $f(X)g(X)$ and $f(X)\otimes U_{|g(X)|}$  are statistically far apart which is necessary for them to be an $\EFI$ pair. 

 H\r{a}stad, Impagliazzo,  Levin and  Luby \cite{HILL} presented a construction of a classical $\PRG$ from a classical $\OWF$. The idea in \cite{HILL} is to first append $HH(X)$ (where $H, H(X)$ are the seed and the output of a $2$-universal hash function based extractor) to $f(X)$ to increase the amount of information about $X$ in $f(X)HH(X)$. This (sort of) makes $X\rightarrow f(X)HH(X)$ an injective function. On appending $HH(X)$, one needs to ensure that the resultant function remains one-way. For this to happen, one can take $|H(X)|$ to be around  $S_2(X|f(X))$  which ensures that $HH(X)$ is nearly independent of $f(X)$. Here $S_\alpha(\cdot)$ represents $\alpha$-R\'enyi entropy (see Definition~\ref{def:renyi}). In \cite{HILL}, $|H(X)|$ depended on the number of preimages of $f(X)$, and hence they needed to condition on the outcome $f(X)$. They then appended a hardcore function $g(X)$ to $f(X)HH(X)$. Doing this maintains computational indistinguishability of $f(X)HH(X)g(X)$ from 
$f(X)HH(X)\otimes U_{|g(X)|}$. Since $f(X)HH(X)$ carries most information about $X$ (injectivity), they argued that $f(X)HH(X)g(X)$ and 
$f(X)HH(X)\otimes U_{|g(X)|}$ are statistically far apart, thus producing an $\EFI$ pair.

 In the quantum regime, this approach runs into an obvious issue that one cannot condition on a quantum state. \cite{KT} circumvent this by measuring the output state using efficient shadow tomography (which preserves expectation values for exponentially many observables). They then get a classical output which allows them to proceed with arguments along the lines of~\cite{HILL}. However they face an additional issue: after conditioning on the measurement outcome, the probability distribution on the preimages may not be uniform. They overcome this using novel arguments. Finally, they show that a pure-state $\OWSG$ implies commitments. Due to the properties needed for known methods for efficient shadow tomography (requiring some norm bounds on states), their argument could only go through for pure-state $\OWSG$s. 
 
Unlike~\cite{HILL} and~\cite{KT}, our proof does not go via measuring the output state of the $\OWSG$ and works for mixed-state $\OWSG$s.  Following~\cite{HILL}, we wish to append $HH^l(X)$ to the combined input-output state $\tau^{XQ^i}$ of an $m$-copy $\OWSG$ for some $i\in [m]$ (here $i$ represents the number of copies of the output state and $l$ represents the output length of the seeded extractor). Consider the states \[
  \tau_0(i,l)=Q^iHH^l(X)Rg(X,R)_\tau,~~
  \tau_1(i,l)=Q^iHH^l(X)_\tau\otimes U_{|R|} \otimes U_{|g(X,R)|}.
 \] To get an EFI pair, we want to argue the following: \[
 \tau_0(i,l) \approx_{\negl_C}\tau_1(i,l);~~
 S(\tau_1(i,l))- S(\tau_0(i,l)) \geq \frac{1}{\poly(n)}.
 \]
The simplified idea for this is as follows. We first identify an $i^*\in [m]$ where \[
\mathrm{D}(\tau^{XQ^{i^*+1}}\|\Id^X\otimes\tau^{Q^{i^*+1}})\approx\mathrm{D}(\tau^{XQ^{i^*}}\|\Id^X\otimes\tau^{Q^{i^*}}).
\] 
     Here $\approx$ means the quantities have a difference at most $O(\log(n))$.
In order to argue the entropy difference between $\tau_1$ and $\tau_0$, we want $
 S(X|Q^{i^*}HH^l(X))_\tau\approx 0,$
which requires $l\approx S(X|Q^{i^*})_\tau$. However, due to the properties of the extractor, we can only extract up to $l= S_2(X|Q^{i^*})_\tau$. Therefore, we want $l\approx S(X|Q^{i^*})_\tau\approx S_2(X|Q^{i^*})_\tau$. From the properties of the quantum hardcore function, we can argue that \[
Q^{i^*}HH^l(X)Rg(X,R)_{\tau}\approx_{\negl_C}Q^{i^*}HH^l(X)_{\tau}\otimes U_{|R|} \otimes U_{|g(X,R)|},
    \] 
    if we can ensure that $XQ^{i^*}HH^l(X)_{\tau}$ is a one-way state (see Definition \ref{def:OWSG}).  For this to hold, we require \[l\approx S_2(X|Q^{i^*+1})_{\tau}\approx S_2(X|Q^{i^*})_{\tau}.\] 
    We identify a substate $\gamma$ of ${\tau}$ (see Definition \ref{def:substate}) with weight at least $\frac{1}{\poly(n)}$ for which 
\begin{align*}
         l_{i^*}\approx S_2(X|Q^{i^*+1})_{\gamma}\approx S_2(X|Q^{i^*})_{\gamma}\approx S(X|Q^{i^*})_{\gamma}.\end{align*}
    We then take a convex mixture of $\tau_0(i,l)$ over $i$ and $l$ to remove the non-uniformity in the construction due to the choice of $i^*$ and $l_{i^*}$. 
    This enables us to create a pair of states $\rho_0$ and $\rho_1$ such that they are computationally indistinguishable and have $1/\poly(n)$ difference in their von Neumann entropies. At this point, $\rho_0$ is uniformly generated, however, $\rho_1$ is still non-uniformly generated. We then take multiple copies of them to amplify their entropy difference and apply a quantum extractor on them to get an imbalanced $\EFI$ pair with one of the outputs being the maximally-mixed state. We then use Fact~\ref{fact:imbalanced_EFI_to_EFI-1}~\cite{KT} that gives an $\EFI$ from an imbalanced $\EFI$, thus completing our proof.

Note that our construction, when applied to the classical case, provides an alternate to the construction provided by~\cite{HILL} to obtain a classical {mildly non-uniform} $\PRG$ from any classical $\OWF$. Since we do not argue conditioned on the output $f(x)$, our construction and analysis is arguably simpler and may be of independent interest. For converting a mildly non-uniform $\PRG$ to a uniform $\PRG$, we can use the same construction as~\cite{HILL}. 
{
 \subsection*{Further work}
We mention some open questions that are related to our work.
\begin{itemize}
 \item Can we get pure-state $\OWSG$s from mixed state $\OWSG$s? In other words, are pure-state $\OWSG$s and mixed-state $\OWSG$s equivalent? Or is there a separation between them?
  \item Construction of (expanding) 1-$\PRSG$ from $\OWSG$ is also an interesting open problem.
 
\end{itemize}
}

\subsection*{Organization}
In Section~\ref{sec:prelim} we introduce some information-theoretic and computation-theoretic preliminaries, definitions, and facts that we will need later for our proofs. In Section~\ref{sec:OWSGtoEFI}, we 
present our construction of an imbalanced $\EFI$ from an $\sOWSG$ and show the reverse direction in Section~\ref{sec:EFItoOWSG}.

\section{Preliminaries} \label{sec:prelim}

In this section, we present some notation, definitions and facts that we will need later for our proofs. For some of the facts, we provide proofs for completeness (although the proofs for these may already exist in the literature). 
\subsection*{Notation}
\begin{itemize}
\item  $[n]$ represents the set $\{1,2,\dots n\}$.
    \item For any $i \in \mathbb{N}$, $Q^i$ represents $Q_{1} Q_2 \dots Q_i$.
    \item $\mathbb{1}(\cdot)$ represents the indicator function.
     \item For $t \in \mathbb{N}$, $U_t$ represents the uniform distribution on $\{0,1\}^t$.
     \item For a register $A$, we let $|A|$ represent the number of qubits in $A$. Similarly, for a binary string $s$, we let $|s|$ represent the number of bits in $s$.
\item $\lambda \in \mathbb{N}$ represents the security parameter and is provided in unary ($1^\lambda$) to algorithms.
 \item By a $\QPT$ algorithm we mean a quantum polynomial-time algorithm. By a non-uniform $\QPT$ algorithm (or adversary), we mean a $\QPT$ algorithm (or adversary) with a polynomial-size advice binary string (in other words a quantum polynomial-size circuit).
    \item If a set of registers (say $A,B,C$) are from the state $\rho$, we represent them as\[
    ABC_\rho =\rho^{ABC}.
    \] 
    \item For a cq-state (with $X$ classical)
    \[\rho^{XQ}  = \sum_{x} \Pr(X=x) \cdot \ketbra{x}{x} \otimes \rho^Q_x,\] 
    and a function $f(x)$, we define 
    \[\rho^{XFQ} \defeq \sum_{x} \Pr(X=x) \cdot \ketbra{x}{x} \otimes \ketbra{f(x)}{f(x)} \otimes \rho^Q_x.\]
    \item $\log(x)$ denotes the binary logarithm of $x$ and $\exp(x) = 2^x$.
    \item $\bern(p)$ denotes the Bernoulli distribution ($p \in [0,1]$).
    \item{Whenever we have a convex combination of states with a different number of qubits, to make the number of qubits equal, we append an appropriate number of $\ket{0}$'s at the end.}
\end{itemize}
\begin{fact}\label{fact:math1}
   For $0\leq x\leq 0.5$, \[\log(1-x) \geq -2x. \]
\end{fact}

\subsubsection*{Information-theoretic preliminaries}
\begin{definition}[$\ell_1$ distance] For an operator $A$, the $\ell_1$ norm is defined as $\|A\|_1 \defeq \Tr \sqrt{A^\dagger A}$. For operators $A,B$, their $\ell_1$ distance is defined as $\onenorm{A}{B}$. We use shorthand $A\approx_{\eps}B$ to denote $\onenorm{A}{B}\leq \eps$.
    
\end{definition}

\begin{definition}[Fidelity]
    For (quantum) states $\rho,\sigma$, \[
    F(\rho,\sigma)\defeq \|\sqrt{\rho}\sqrt{\sigma}\|_1.
    \]
\end{definition}
\begin{definition}[Bures metric]
    For states $\rho,\sigma$, \[
    \Delta_B(\rho,\sigma)\defeq \sqrt{1-F(\rho,\sigma)}.
    \]
   
\end{definition}

\begin{fact}[\cite{fuchs}]\label{fact:fuchs}
    For states $\rho, \sigma$,
\[ \Delta_B^2(\rho, \sigma) \leq \frac{\onenorm{\rho}{ \sigma}}{2} \leq
\sqrt{2}\Delta_B(\rho, \sigma).\]
\end{fact}
\begin{definition}For a  state $\rho$, its von Neumann entropy is defined as \[
S(\rho)\defeq -\Tr(\rho\log\rho).
\]
\end{definition}
\begin{fact}[Chain-rule for entropy] \label{fact:chainentropy} For a state $\rho_{AB}$, 
    \[S(AB)_\rho = S(A)_\rho + S(B|A)_\rho.\]
\end{fact}
\begin{fact}\label{fact:entropy_inequalities} For  a cq-state  $\rho^{XQ}$ (with $X$ classical), 
\[S(Q|X)_\rho \leq  S(Q)_\rho\leq S(XQ)_\rho \leq S(Q)_\rho + |X|.\] 
\end{fact}
\begin{proof}
    The first inequality follows from Fact~\ref{fact:alpha_renyi_inequalities}. For the second inequality consider,
\begin{align*}
       S(XQ)_\rho&=S(Q)_\rho+S(X|Q)_\rho &\mbox{(Fact~\ref{fact:chainentropy})}
       \\&\geq S(Q)_\rho.  &\mbox{(Fact \ref{fact:renyi_classical})}
    \end{align*}
The last inequality follows from Fact~\ref{fact:chainentropy} and Fact~\ref{fact:alpha_renyi_inequalities}.
\end{proof}
\begin{definition}[R\'enyi entropy]\label{def:renyi}  For a  state $\rho$, the R\'enyi entropy of order $\alpha \in (0, 1) \cup (1,\infty)$ is given as
\[S_\alpha(\rho) \defeq \frac{1}{1-\alpha}
\log \Tr(\rho^\alpha).\]
$S_0,S_1$ and $S_\infty$ are defined as limits of $S_\alpha$ for $\alpha \rightarrow \{0,1,\infty\}$. In particular, $\lim_{\alpha\rightarrow 1} S_\alpha(\rho)=S(\rho)$.  

If $\{p_i\}_i$ are the eigenvalues of $\rho$, then the  R\'enyi entropy reduces to 
\begin{equation}
    \label{eq:renti_definition}
S_\alpha(\rho)= \frac{1}{1-\alpha}\log\left(\sum_ip_i^\alpha\right).\end{equation}

    \end{definition}

    \begin{definition}[Conditional R\'enyi entropy]\label{def:renyi_conditional}  For  states $\rho$ and $\sigma$, the "sandwiched" $\alpha$-R\'enyi divergence for $\alpha \in (0, 1) \cup (1,\infty)$ is given as
    \[\mathrm{D}_\alpha(\rho\|\sigma)\defeq\frac{1}{\alpha-1}\log\left(\Tr\left[(\sigma^{\frac{1-\alpha}{2\alpha}}\rho\sigma^{\frac{1-\alpha}{2\alpha}})^\alpha\right]\right).
    \]
The corresponding $\alpha$-conditional R\'enyi entropy is given by\[
        S_\alpha(A|B)_\rho \defeq -\min_{\sigma^B} \mathrm{D}_\alpha(\rho^{AB}\|\Id^A\otimes\sigma^B).
        \]
$\mathrm{D}_0,\mathrm{D}_1$ and $\mathrm{D}_\infty$ are defined as limits of $\mathrm{D}_\alpha$ for $\alpha \rightarrow \{0,1,\infty\}$. 
    \end{definition}
    \begin{fact}[\cite{marco_book}]\label{fact:H_infinity}
   The sandwiched $\alpha$-R\'enyi divergence as $\alpha\rightarrow\infty$ becomes  \[\mathrm{D}_\infty(\rho\|\sigma)= \min\{\lambda : \rho
   \leq 2^\lambda\sigma\}.\]  
    \end{fact}
  { \begin{fact}[\cite{renyi_0,renyi_0_new}]\label{fact:D_0}
   The sandwiched $\alpha$ -R'enyi divergence satisfies \[\mathrm{D}_1(\rho\|\sigma)\geq -\log\Tr\left(\Pi_{\support(\rho)} \sigma\right) .\] \end{fact}}
 {\begin{definition}[Smooth R\'enyi entropy, \cite{renner}]  
 Let $\rho$ be a state. Let us denote the $\eps$-ball around $\rho$ as:
 \[
 \mathcal{B}^\eps(\rho)\defeq \{\widetilde{\rho} : \widetilde{\rho} \geq 0, \Delta_B(\widetilde{\rho},\rho)\leq \eps, \Tr(\widetilde{\rho})\leq 1\}.
 \]For any $\eps \geq 0$, the $\eps$-smooth min and max-entropies 
 are defined as
\[
S_\infty^\eps(A|B)_\rho\defeq\sup_{\rho^\prime\in \mathcal{B}^\eps(\rho)} S_\infty(A|B)_{\rho^\prime},
\]
\[
S_0^\eps(A|B)_\rho\defeq\inf_{\rho^\prime\in \mathcal{B}^\eps(\rho)} S_0(A|B)_{\rho^\prime}.
\]
 
 \end{definition}}
    
\begin{fact}[Data processing] \label{fact:data} Let $\rho, \sigma$ be states and $\Phi$ be a CPTP map. Then, 
  \[\Delta_B(\Phi(\rho), \Phi(\sigma)) \leq  \Delta_B(\rho, \sigma).\]
  The above also holds for the $\ell_1$ distance. For all $\alpha \geq \frac{1}{2}$~\cite{marco_book},
  \[\mathrm{D}_\alpha(\Phi(\rho)\|\Phi(\sigma)) \leq \mathrm{D}_\alpha(\rho\|\sigma).\]
\end{fact}

\begin{fact}[\cite{khatri2024principles}] \label{fact:additivity} For all $\alpha \in [\frac{1}{2}, 1) \cup (1,\infty)$ and states $\rho_1,\rho_2,\sigma$,
\[
\mathrm{D}_\alpha(\rho_1\otimes \sigma\|\rho_2\otimes \sigma)=
\mathrm{D}_\alpha(\rho_1\|\rho_2).
\]
\end{fact}
\begin{fact}[Non-negativity, \cite{marco_book}] \label{fact:non_negative} For all states $\rho$ and $\sigma$, \[
 \mathrm{D}(\rho\|\sigma)\geq 0.
\]\vspace{-6mm}
\begin{fact}[Joint-convexity of  relative entropy]\label{fact:joint_convexity} Let $\rho_0, \rho_1, \sigma_0, \sigma_1 $ be states and $\lambda\in[0,1]$. Then \[
\diver{\lambda\rho_0+(1-\lambda)\rho_1}{\lambda\sigma_0+(1-\lambda)\sigma_1} \leq \lambda\diver{\rho_0}{\sigma_0}+(1-\lambda)\diver{\rho_1}{\sigma_1}.
\]
\end{fact}
\begin{fact}[Chain-rule for relative entropy]\label{fact:chain_relative} Let $\rho=\sum_x\mu_x\ketbra{x}{x}\otimes\rho_x$ and
$\rho^1=\sum_x\mu^1_x\ketbra{x}{x}\otimes\rho^1_x$.

\noindent Then,
    \[
    \diver{\rho^1}{\rho}=\diver{\mu^1}{\mu}+\E_{x\leftarrow\mu^1}\diver{\rho^1_x}{\rho_x}.
    \]
\end{fact}
    
\end{fact}
\begin{fact} [Monotonicity~\cite{marco_book}]\label{fact:monotonicity_renyi}
  $\alpha$-R\'enyi divergence is monotonically increasing  while $\alpha$-R\'enyi entropy (and  $\alpha$-conditional R\'enyi entropy) is monotonically decreasing with increasing $\alpha$.
\end{fact}
\begin{fact}\label{fact:conjugation}
Let $A\geq 0$ be a positive semi-definite matrix. Then for any matrix $C$ ({for which the product below is defined}), \[
 CAC^\dagger\geq 0.
\]
In particular, if $X\geq Y$, then \[
CXC^\dagger\geq CYC^\dagger.
\]
\end{fact}

\begin{fact}[\cite{marco_book}, eqs. (5.41) and (5.96)]\label{fact:alpha_renyi_inequalities}
For a cq-state $\rho^{ABC}$ (with $C$ classical) and all $\alpha\in[\frac{1}{2},\infty]$, 
\[S_\alpha(A|B)_\rho-|C|\leq S_\alpha(A|BC)_\rho\leq S_\alpha(A|B)_\rho \leq |A|.\]
Note: $C$ being classical is needed only for the first inequality. 
    
\end{fact}
\begin{fact}[\cite{marco_book}]\label{fact:renyi_expression_classical}
Let $\rho^{XAB}= \sum_x p_x \cdot  \ketbra{x}{x}^X \otimes \rho^{AB}_x$ be a cq-state. For all $\alpha \in (0, 1) \cup (1,\infty)$, \[
S_\alpha(A|BX)_\rho=\frac{\alpha}{1-\alpha}\log\left(\sum_xp_x \exp{\left(\frac{1-\alpha}{\alpha}S_\alpha(A|B)_{\rho_x}\right)}\right).
\]
    
\end{fact}

\begin{fact}[\cite{Berta_2011} and \cite{BJL_2023}]\label{fact:identity_upper} For a state $\sigma_{AB}$,
\[\sigma_{AB}\leq 2^{|B|}(\sigma_A\otimes \Id_B).\]
For a cq-state $\sigma_{XB}$ (with $X$ classical),
\[\sigma_{XB}\leq \sigma_X \otimes \Id_B,\]
\[\sigma_{XB}\leq \Id_X \otimes \sigma_B.\]
        
\end{fact}

\begin{fact} \label{fact:renyi_classical} Let $\rho^{XQ}$ be a cq-state (with $X$ classical). Then $ S_\alpha(X|Q)_\rho\geq 0$ for $\alpha \in [0,\infty]$. 
\end{fact}
\begin{proof}Consider,
   \begin{align*}
       S_\alpha(X|Q)_\rho &\geq S_\infty(X|Q)_\rho \nonumber&\mbox{(Fact \ref{fact:monotonicity_renyi})}
       \\
       &=- \min_{\sigma_B} \mathrm{D}_\infty(\rho_{XQ}\|\Id_X\otimes\sigma_Q)  \nonumber
       \\
       &\geq -\mathrm{D}_\infty(\rho_{XQ}\|\Id_X\otimes\rho_Q) \\
       &  \geq 0. \nonumber & \mbox{(Facts~\ref{fact:H_infinity} and~\ref{fact:identity_upper})} 
   \end{align*}
 \end{proof}
\begin{definition}[Flat state]
A distribution $P$ is said to be flat if all non-zero probabilities are the same. A quantum state $\rho$ is said to be flat if the probability distribution formed by its eigenvalues is flat. From Definition \ref{def:renyi}, we see that for a flat state $\rho$ and any $\alpha,\beta\in[0,\infty]$,  $S_\alpha(\rho)=S_\beta(\rho)$. 
\end{definition}

\begin{definition}[$2$-flat state]
Let $p_{max}$ and $p_{min}$ be the maximum and minimum non-zero probability values of the distribution $P$. $P$ is said to be a $2$-flat distribution  if \[\frac{p_{max}}{p_{min}}\leq 2.\]
     A quantum state $\rho$ is said to be $2$-flat if the probability distribution formed by its eigenvalues is $2$-flat.
     From Definition \ref{def:renyi}, we see that for a state $\rho$, \[
    S_{\infty}(\rho) = \log\left(\frac{1}{p_{max}}\right) \text{ and } S_{0}(\rho)\leq \log\left(\frac{1}{p_{min}}\right).\] 
 From Fact \ref{fact:monotonicity_renyi}, we see that for a $2$-flat state $\rho$, for all $\alpha,\beta\in[0,\infty],$
 \begin{equation}
     \label{eq:2flat_state}
   |S_\alpha(\rho)-S_\beta(\rho)| \leq S_{0}(\rho)- S_{\infty}(\rho) \leq \log \left(\frac{p_{max}}{p_{min}}\right)\leq 1.
 \end{equation}
\end{definition}
\begin{definition}\label{def:substate}
    Let $\rho,\rho_1$ and $\rho_2$ be states such that \[
    \rho=p\cdot\rho_1+(1-p)\cdot \rho_2.    \]
    We say that $\rho_1$ is a substate of $\rho$ with probability atleast $p$.
\end{definition}

\begin{fact}[Fannes' inequality~\cite{Nielsen_Chuang}] \label{fact:fannes} For states $\rho$ and $\sigma$, \[
\vert S(\rho) -S(\sigma)\vert \leq \log(d)\onenorm{\rho}{\sigma}+\frac{1}{e}.
\] Here $d$ is the dimension of the Hilbert space of $\rho$ (and $\sigma$).
\end{fact}
\begin{fact}[Extractor~\cite{renner}] \label{fact:hashing}  Let $X\in\{0,1\}^n$ and 
\[\tau^{XQ} = \sum_{x} \Pr(X=x)  \cdot 
 \ketbra{x}{x}^X \otimes \tau^Q_x,\] be a cq-state. Let $l=S_2(X|Q)_\tau$ and $s = O(n)$. There exists extractor 
\[
\ext: \{0,1\}^n \times \{0,1\}^s \rightarrow \{0,1\}^l,\]
such that for all $l^\prime \leq l$,
\[
QHH^{l^\prime}(X)_\tau\approx_{2^{-(l-l^\prime)/2}} Q_\tau\otimes U_s\otimes U_{l^\prime}
\]    
where  
\[\tau^{XHH(X)Q} = \sum_{x, h} \Pr(X=x) \cdot 2^{-s} \cdot 
 \ketbra{x,h,\ext(x,h)}{x,h,\ext(x,h)} \otimes \tau^Q_x,\] 
and $H^{l^\prime}(X)$ represents the $l^\prime$-bit prefix of $H(X)$.
\end{fact}
\begin{fact}[Uhlmann’s Theorem~\cite{uhlmann}]\label{fact:uhlmann}
    Let $\rho_A, \sigma_A$ be states. Let $\rho_{AB}$ be a purification
of $\rho_A$ and $\sigma_{AC} $ be a purification of $\sigma_A$. There exists an isometry $V$ such that,
\[\Delta_B (\ketbra{\rho}{\rho}_{AB}, \ketbra{\theta}{\theta}_{AB}) = \Delta_B(\rho_A, \sigma_A),\]
where $\ket{\theta}_{AB} = (\Id_A \otimes V )\ket{\sigma}_{AC}$.
\end{fact}

\begin{fact}[Distance preserving extension]\label{fact:extension}
    Let $\rho_{AB}, \sigma_A$ be states. There exists an extension $\theta_{AB}$ of $\sigma_A$, that is $\Tr_B \theta_{AB} = \sigma_A$, such that,
\[\Delta_B (\rho_{AB}, \theta_{AB}) = \Delta_B(\rho_A, \sigma_A).\]
\end{fact}
\begin{proof}
Let $\rho_{ABR}$ be a purification of $\rho_{AB}$. Let $\sigma_{AC}$ be a purification of $\sigma_A$.  From Uhlmann's Theorem (Fact \ref{fact:uhlmann}), there exists an isometry $V$ such that $\ket{\theta}_{ABR} = (\Id_A \otimes V )\ket{\sigma}_{AC}$ and 
\[\Delta_B (\rho_{ABR}, \theta_{ABR}) = \Delta_B(\rho_A, \sigma_A).\]
From the data-processing inequality (Fact~\ref{fact:data}), we get 
\[\Delta_B(\rho_A, \theta_A) \leq \Delta_B (\rho_{AB}, \theta_{AB})  \leq \Delta_B (\rho_{ABR}, \theta_{ABR}) = \Delta_B(\rho_A, \sigma_A).\]
This shows the desired by noting that $\theta_A = \sigma_A$.
\end{proof}

\begin{fact}[\cite{Tomamichel_2013}] \label{fact:smooth_entopy_multiple_copies}
    For a state $\rho$ on $n$-qubits, 
\[S_0^\eps(\rho^{\otimes t})\leq t\cdot S(\rho)+O(\sqrt{tn})\Phi^{-1}(\eps) +O(\log(n)).\] 
\[S_\infty^\eps({\rho}^{\otimes t})\geq t\cdot S({\rho})-O(\sqrt{tn})\Phi^{-1}(\eps) +O(\log(n)).\] 
Here $\Phi(x)=\int_{-\infty}^x\frac{e^{-t^2/2}}{\sqrt{2\pi}}dt.$
\end{fact}
\begin{fact}[\cite{alphabit}]\label{fact:alphabit}
    For the function  $\Phi(x)=\int_{-\infty}^x\frac{e^{-t^2/2}}{\sqrt{2\pi}}dt$, and $\eps\leq \frac{1}{2}$,\[
    |\Phi^{-1}(\eps)|\leq \sqrt{2\log\frac{1}{2\eps}}.
    \]
\end{fact}

\begin{fact}[Quantum extractor~\cite{AJ22}, Theorem 2]\label{fact:quantum_extractor} Let $\psi_R$ be a state on $n$-qubits, $\eps>0$ and $s=4n-S_\infty^{\eps}(\psi_R)+O(\log\left(\frac{1}{\eps}\right))$.  There exists a unitary $G$ such that (below $S$ represents a register on $s-3n$ qubits),
\[
\onenorm{\Tr_S G(\psi_R\otimes \ketbra{0}{0}\otimes U_{s})G^\dagger}{U_{4n+1}}\leq \eps.
\]
This gives a quantum extractor with seed length $s$ and output length $4n+1$:
\[\ext_Q(\psi_R,U_{s})= \Tr_S G(\psi_R\otimes \ketbra{0}{0}\otimes U_{s})G^\dagger.\]
In addition, for any $\delta>0$, \[
S_0^\delta(\ext_Q(\psi_R,U_{s}))\leq S^\delta_0(\psi_R) + s.
\]
\end{fact}
\begin{fact} \label{fact:convex_entropy_diff}
 Let $\rho_0$, $\rho_1$ and $\rho_2$ be states. Let \[
    \sigma_0=p\rho_0+(1-p)\rho_2~,~  \sigma_1=p\rho_1+(1-p)\rho_2.
    \] 
    For any $p\in [0,\frac{1}{2}]$,  \[
    S(\sigma_{1})-  S(\sigma_{0})\geq p\cdot\left(S(\rho_1)-S(\rho_0)-\log\left(\frac{1}{p}\right)-2\right).
    \]
\end{fact}
\begin{proof}
Let $Q$ denote the quantum register corresponding to the states $\rho_0, \rho_1$ and $\rho_2$.
Consider states 
    \[
{\kappa}^{XQ}_0=p\cdot\ketbra{0}{0}^X\otimes \rho_0^{Q}+(1-p)\cdot\ketbra{1}{1}^X\otimes \rho_2^{Q}.
    \]
    \[
    {\kappa}^{XQ}_1=p\cdot\ketbra{0}{0}^X\otimes \rho_1^{Q}+(1-p)\cdot\ketbra{1}{1}^X\otimes \rho_2^{Q}.
    \]
From the above construction, we have \begin{align}
    S(\sigma_{0})&=S(p\cdot \rho_0^{Q}+(1-p)\cdot\rho_2^{Q}) \nonumber
\\&=S(Q)_{{\kappa}_0}\nonumber
    \\&\leq S(XQ)_{{\kappa}_0}&\mbox{(Fact \ref{fact:entropy_inequalities})}\nonumber
     \\&= S(X)_{{\kappa}_0}+S(Q|X)_{{\kappa}_0}&\mbox{(Fact \ref{fact:chainentropy})}\nonumber
     \\&= S(\bern(p))+p\cdot S(\rho_0)+(1-p)\cdot S(\rho_2).\label{eq:ent_1}
\end{align}
Similarly, we get
\begin{align}
    S(\sigma_{1})&=S(p\cdot \rho_1^{{Q}}+(1-p)\cdot\rho_2^{{Q}}) \nonumber
    \\&=S({Q})_{{\kappa}_1}\nonumber
    \\&\geq S({Q}|X)_{{\kappa}_1}&\mbox{(Fact \ref{fact:entropy_inequalities})}\nonumber
     \\&=p\cdot S(\rho_1)+(1-p)\cdot S(\rho_2).\label{eq:ent_2}
\end{align}
Consider,\begin{align}
    S(\bern(p))&=p\log(\frac{1}{p})+(1-p)\log(\frac{1}{1-p})\nonumber
    \\&\leq p\log(\frac{1}{p})+\log(\frac{1}{1-p})\nonumber
    \\&\leq p\log(\frac{1}{p})+\log(1+2p) &\mbox{(for $p\in[0,\frac{1}{2}], (1-p)(1+2p)\geq 1$)} \nonumber 
    \\&\leq p\log(\frac{1}{p})+2p.  &\mbox{($1+x\leq e^x$)} \label{eq:ent_3}
\end{align}
From eqs. \ref{eq:ent_1} and \ref{eq:ent_2}, we get 
\small\begin{align*}
    S(\sigma_{1})-  S(\sigma_{0})&\geq p\cdot(S(\rho_1)-S(\rho_0))-S(\bern(p))
    \\&\geq p\cdot(S(\rho_1)-S(\rho_0)) - p\log\left(\frac{1}{p}\right)-2p &\mbox{(eq. \ref{eq:ent_3})}
    \\&=p\cdot\left(S(\rho_1)-S(\rho_0)-\log\left(\frac{1}{p}\right)-2\right). \qedhere
\end{align*}

\end{proof}
\begin{fact}\label{fact:trace_with_identity} For quantum states $\rho^{AB}$ and $\sigma^A$,
    \[
    \Tr(\rho^{AB}(\sigma^A\otimes\Id^B))=\Tr(\rho^A\sigma^A).
    \]
\end{fact}
\begin{fact}\label{fact:quantum_chaining} For a quantum state $\sigma^A$ and a cq-state $\rho^{AJ}$ with $J$ classical,
    \[
 {\mathrm{D}(\rho^A\|\sigma^A)\leq}\E_{j\leftarrow\rho^J} \left[\mathrm{D}(\rho_j^A\|\sigma^A)\right]\leq \mathrm{D}(\rho^A\|\sigma^A)  + |J|.
\]
\end{fact}
\begin{proof}
The first inequality follows from Fact \ref{fact:joint_convexity}. For the other inequality, consider
\begin{align*}
     \E_{j\leftarrow\rho^J} \left[\mathrm{D}(\rho_j^A\|\sigma^A)\right] &\leq \mathrm{D}(\rho^{J}\| U^J) + \E_{j\leftarrow\rho^J} \left[\mathrm{D}(\rho_j^A\|\sigma^A)\right] &\mbox{(Fact \ref{fact:non_negative})} \\& = \mathrm{D}(\rho^{AJ}\|\sigma^A\otimes U^J) &\mbox{(Fact \ref{fact:chain_relative})}
     \\&= \Tr(\rho^{AJ}\log(\rho^{AJ}))-\Tr(\rho^{AJ}\log(\sigma^A\otimes U^J))
    \\&= \Tr(\rho^{AJ}\log(\rho^{AJ}))-\Tr(\rho^{AJ}\log(\sigma^A\otimes \Id^J))+|J|
     \\&\leq  \Tr(\rho^{AJ}\log(\rho^{A}\otimes\Id^J))-\Tr(\rho^{AJ}\log(\sigma^A\otimes \Id^J))+|J| &\mbox{(Fact \ref{fact:identity_upper})}  \\&=  \Tr(\rho^{AJ}(\log(\rho^{A})\otimes\Id^J))-\Tr(\rho^{AJ}(\log(\sigma^A)\otimes \Id^J))+|J| 
     \\&= \Tr(\rho^{A}\log(\rho^{A}))-\Tr(\rho^{A}\log(\sigma^A))+|J| &\mbox{(Fact \ref{fact:trace_with_identity})}
      \\&= \mathrm{D}(\rho^{A}\|\sigma^A)+|J|. \qedhere
\end{align*}    
\end{proof}
\subsubsection*{Computation-theoretic preliminaries}
\begin{definition}[$\negl$ function]\label{def:negligible}
    We call a function $\mu : \mathbb{N} \rightarrow \mathbb{R}^{+}$ negligible if for every positive polynomial $p(\cdot)$, there exists an $N \in \mathbb{N}$ such that for all $\lambda > N$, $\mu(\lambda)< \frac{1}{p(\lambda)}$. We let $\mu(\lambda) = \negl(\lambda)$ represent that $\mu(\lambda)$ is a negligible function in $\lambda$. 
\end{definition}

\begin{definition}[noticeable function]\label{def:noticeable}    We call a function $\mu : \mathbb{N} \rightarrow \mathbb{R}^{+}$ noticeable if there exists a positive polynomial $p(\cdot)$ and an $N \in \mathbb{N}$ such that for all $\lambda > N$, $\mu(\lambda) > \frac{1}{p(\lambda)}$.  
\end{definition}

\begin{definition}[$\poly$]\label{def:poly}  We let $\poly(\lambda)$ denote the class of all positive polynomials in $\lambda$.
\end{definition}

\begin{definition}[Computationally indistinguishable] We say two ensemble of states $\{\rho_\lambda, \sigma_\lambda\}_{\lambda \in \mathbb{N}}$ are computationally indistinguishable, denoted as $\rho \approx_{\negl_C}\sigma$, if for every non-uniform $\QPT$ adversary $A$,
    \[\Big\vert\Pr[A(1^\lambda,\rho_\lambda)=1]-\Pr[A(1^\lambda,\sigma_\lambda)=1] \Big\vert = \negl(\lambda).\] 
    
\end{definition}
Morimae and Yamakawa~\cite{owsg} have defined a quantum one-way state generator. We use the following definition which differs from theirs in two points. The algorithm $\ver$ and the adversary 
 A in their definition are $\QPT$ algorithms, whereas in our definition $\ver$ is an unbounded algorithm and adversary 
 A is a non-uniform $\QPT$ algorithm. 
\begin{definition}[{Statistically-verifiable} quantum one-way state generator] \label{def:OWSG} An $m$-copy {statistically-verifiable} quantum one-way state generator ($\sOWSG$) is a set of algorithms: 
\begin{enumerate}
\item $\keygen(1^\lambda)\rightarrow x$: It is a $\QPT$ algorithm that on input the security parameter, outputs a classical key $x\in\{0,1\}^n$. 
    \item $\gen(1^\lambda, x)\rightarrow \phi_x$: It is a $\QPT$ algorithm that takes as input the security parameter and $x \in \{0,1\}^n$ and outputs a ({mixed}) quantum state $\phi_x$. In other words, $\gen(1^\lambda,x)=\phi_x^{}$ .
    \item $\ver(x^\prime, \phi_x)\rightarrow \tang /\bot:$ It is an algorithm that on input a string $x^\prime$ and $\phi_x$  gives as output $\tang$ or $\bot$.

\end{enumerate}
We require the following correctness and security conditions to hold. 

\noindent {\bf Correctness:} 
\[\Pr[\tang\leftarrow\ver(x,\phi_x):x\leftarrow\keygen(1^\lambda), \phi_x\leftarrow\gen(1^\lambda,x)]\geq 1-\negl(\lambda).\] 

\noindent{\bf Security:}

\noindent For any non-uniform $\QPT$ adversary $A$ (with running time in $\poly(\lambda)$),  
\[\Pr[\tang\leftarrow\ver(x^\prime,\phi_x):x\leftarrow\keygen(1^\lambda), \phi_x\leftarrow\gen(1^\lambda,x),x^\prime\leftarrow A(1^\lambda, \phi_x^{\otimes m})]= \negl(\lambda).\] 
We call the state 
\begin{equation}\label{eq:one-way-state}
    \tau^{XQ^m}=\sum_{x\in\{0,1\}^n}\Pr(X=x)_\tau\cdot\ketbra{x}{x}\otimes \phi_x^{\otimes m},
\end{equation} 
an $m$-copy one-way state and the function $x\rightarrow \phi_x^{\otimes m}$  an $m$-copy quantum one-way function. For a non-uniform $\QPT$ adversary $A$, security implies that $\forall i\in[m]$,  \[
\Pr[\tang \leftarrow \ver(A(1^\lambda,\tau^{Q^i}), \tau^{Q_{i+1}} )]= \negl(\lambda).
\]

\end{definition}

\begin{definition}[EFI~\cite{brakerski2022computational}]
An $\EFI$  is a $\QPT$ algorithm
$\gen(b,1^\lambda)\rightarrow\rho_b$ that, on input
$b\in \{0,1\}$ and the security parameter $\lambda$, outputs a quantum state $\rho_b$
such that the following conditions are satisfied.
\begin{enumerate}
    
    \item $\rho_0\approx_{\negl_C}\rho_1$. 
    \item $\rho_0$ and $\rho_1$ are statistically distinguishable with noticeable advantage. In other words, $\frac{1}{2}\onenorm{\rho_0}{\rho_1}$ is a noticeable function in $\lambda$.
\end{enumerate}
We call $(\rho_0,\rho_1)$ an $\EFI$ pair.
\end{definition}
\begin{fact}[\cite{Morimae_2022}]\label{fact:EFI_amplification}
    $\EFI$  with noticeable advantage in statistical distinguishability are equivalent to  $\EFI$  with $1-\negl(\lambda)$ advantage. 
\end{fact}

\begin{definition}[$s^*$-imbalanced EFI~\cite{KT}] \label{def:imEFI}Let $s^*(\lambda): \mathbb{N} \rightarrow \mathbb{N}$ be a function. An $s^*$-imbalanced $\EFI$  is a $\QPT$ algorithm
EFI$_s(1^\lambda
, b) \rightarrow \rho_b(s)$ that obtains advice string $s$ and, on input $b \in \{0, 1\}$ and
security parameter $\lambda$, outputs a state $\rho_b(s)$ such that:
\begin{enumerate}
    
    \item For all $s\leq s^*(\lambda)$: $\rho_0(s) \approx_{\negl_C} \rho_1(s)$. 
    \item For all $s\geq s^*(\lambda)$: $\rho_0(s)$ and $\rho_1(s)$ are statistically distinguishable with noticeable advantage. That is  $\frac{1}{2}\onenorm{\rho_0(s)}{\rho_1(s)}$ is a noticeable function in $\lambda$.
\end{enumerate}

\end{definition}
\begin{fact}[\cite{KT}]\label{fact:imbalanced_EFI_to_EFI}
Let $s^*(\lambda) \in \poly(\lambda)$. Then $s^*$-imbalanced $\EFI$  imply quantum commitments. 
\end{fact}

\begin{fact}[\cite{brakerski2022computational}]
    Quantum commitments are equivalent to $\EFI$.
\end{fact}

\begin{fact}[\cite{Goldreich_2001}]\label{fact:hardcore_classical}
Let $f$ be a classical one-way function. Define $h(x,r)=(f(x),r)$ where $x\in\{0,1\}^n$ and $r\in\{0,1\}^{2n}$. For any constant $c>0$, there exists an efficiently computable hardcore function $g(x,r)$, with $|g(x,r)|=c\log(n)$, such that $g(x,r)$ is a hardcore function of $h$. In other words, 
\[
f(U_n) R ~g(U_n,R)\approx_{\negl_C} f(U_n)\otimes U_{2n}\otimes U_{c\log(n)}.
\]
    \end{fact}

   \begin{fact}[Quantum hardcore function]\label{fact:quantum_hardcore}
    Let $X \in \{0,1\}^n, m \in \poly(n)$ and $\tau^{XQ^m}$ be an $m$-copy one-way state (see Definition~\ref{def:OWSG}). For any constant $c>0$, there exists an efficiently computable hardcore function $g(x,r)$ for $|r|=2n$ and $|g(x,r)|=c\log(n)$ such that \[
Q^m R~ g(X,R)_\tau\approx_{\negl_C} Q^m_\tau\otimes U_{2n}\otimes U_{c\log(n)}.
\]
\end{fact}
\begin{proof}(Sketch)
    The proof of Fact \ref{fact:hardcore_classical} reduces to the existence of hardcore predicates for classical one-way functions.  Goldreich-Levin with quantum advice (\cite{quantum-Goldreich-Levin}) gives the existence of hardcore predicate for quantum one-way functions also and hence, the proof of this claim follows along similar lines as Fact \ref{fact:hardcore_classical}.  
\end{proof}
  
\begin{fact} \label{fact:indistinguishability_multiple_samples}
    Let $\rho_0$ and $\rho_1$ be states generatable by a non-uniform $\QPT$ algorithm. If  $\rho_0\approx_{\negl_C}\rho_1$, then $\rho_0^{\otimes t}\approx_{\negl_C}\rho_1^{\otimes t}$ for any $t \in \poly(\lambda)$ where  $\lambda$ is the security parameter.
\end{fact}
\begin{proof}
This follows using a standard hybrid argument. Assume for contradiction that there exists a non-uniform $\QPT$ adversary $A$ such that for infinitely many $\lambda$ and a positive polynomial $p(\cdot)$, 
  \[\Big\vert\Pr[A(1^\lambda,\rho_0^{\otimes t})=1]-\Pr[A(1^\lambda,\rho_1^{\otimes t})=1] \Big\vert \geq \frac{1}{p(\lambda)}.\] 
We define a series of hybrids $H_0, H_1 \dots H_t$ where $H_i=\rho_0^{\otimes t-i}\otimes\rho_1^{\otimes i}$. Then \begin{align*}
       \frac{1}{p(\lambda)} &\leq \Big\vert\Pr[A(1^\lambda,\rho_0^{\otimes t})=1]-\Pr[A(1^\lambda,\rho_1^{\otimes t})=1]\Big\vert \\&=\Big\vert\Pr[A(1^\lambda,H_0)=1]-\Pr[A(1^\lambda,H_t)=1]\Big\vert
       \\&\leq \sum_{i\in[t]}\Big\vert\Pr[A(1^\lambda,H_{i-1})=1]-\Pr[A(1^\lambda,H_i)=1]\Big\vert. &\mbox{(By triangle-inequality)} 
  \end{align*}
Thus, there exists a $j\in[t]$ for which \begin{equation}
    \label{eq:hybrids}
\Big\vert\Pr[A(1^\lambda,H_{j-1})=1]-\Pr[A(1^\lambda,H_j)=1]\Big\vert\geq \frac{1}{p(\lambda)\cdot t}\geq \frac{1}{q(\lambda)},\end{equation} for some positive polynomial $q(\cdot)$ because $t\in\poly(\lambda)$.

  Using non-uniform $\QPT$ adversary $A$, we now create a non-uniform adversary $A^\prime$ to distinguish between $\rho_0$ and $\rho_1$ as follows:
  
  \textbf{$A^\prime(1^\lambda, \rho_c)$:}
  \begin{itemize}
  \item $A^\prime$  takes as input  a state $\rho_c$ for $c\in\{0,1\}$ along with the security parameter $\lambda$.
      \item $A^\prime$  takes as (non-uniform) advice $j\in [t]$. 
      \item $A^\prime$ returns $A(1^\lambda,\rho_0^{\otimes t-j}\otimes \rho_c\otimes \rho_1^{\otimes j-1})$.
  \end{itemize} 
 Since $\rho_0$ and $\rho_1$ are generatable by non-uniform $\QPT$ algorithms, adversary $A^\prime$ is also non-uniform $\QPT$. From eq. \ref{eq:hybrids}, we get that for infinitely many $\lambda$
  \[
   \Big\vert\Pr[A^\prime(1^\lambda,\rho_0)=1]-\Pr[A^\prime(1^\lambda,\rho_1)=1]\Big\vert=\Big\vert\Pr[A(1^\lambda,H_{j-1})=1]-\Pr[A(1^\lambda,H_j)=1]\Big\vert\geq \frac{1}{q(\lambda)}.
  \]
This contradicts the fact that $\rho_0\approx_{\negl_C}\rho_1$. Hence, the desired statement follows.   
\end{proof}
{
\begin{fact} \label{fact:indistinguishability_convex_combination}
    Let $\sigma$ be a state generatable by a non-uniform $\QPT$ algorithm. Let $\rho_0$ and $\rho_1$ be states such that $\rho_0\approx_{\negl_C}\rho_1$. Let \[
    \kappa_0=p\sigma+(1-p)\rho_0~,~  \kappa_1=p\sigma+(1-p)\rho_1.
    \] 
    For any $p\in [0,1]$,  $\kappa_0 \approx_{\negl_C} \kappa_1.$
\end{fact}}
\begin{proof}
 Assume for contradiction that there exists a non-uniform $\QPT$ adversary $A$ such that for infinitely many $\lambda$ and a positive polynomial $q(\cdot)$, 
  \begin{equation}
      \label{eq:indis_kappa}\Big\vert\Pr[A(1^\lambda,\kappa_0)=1]-\Pr[A(1^\lambda,\kappa_1)=1] \Big\vert \geq \frac{1}{q(\lambda)}. \end{equation} 
We now create a non-uniform adversary $A^\prime$ as follows:
  
  \textbf{$A^\prime(1^\lambda, \rho_c)$:}
  \begin{itemize}
  \item $A^\prime$  takes as input  a state $\rho_c$ for $c\in\{0,1\}$ along with the security parameter $\lambda$.
  \item $A^\prime$ receives (non-uniform) advice $p$ in the form of its binary encoding in polynomially many bits (in $\lambda$)\footnote{The exponentially small error in the specification of $p$ can be easily incorporated in the proof, however, we do not do that explicitly for the ease of presentation.}.
       \item {$A^\prime$ creates the state $\tau=p\sigma+(1-p)\rho_c=\kappa_c$}
      and returns $A(1^\lambda,\tau)$\footnote{This is equivalent to creating the state $\sigma$ in another register, tossing a $p$-biased coin and outputting $A(1^\lambda,\sigma)$  or $A(1^\lambda,\rho_c)$ depending on the outcome of the coin.}.
  \end{itemize} 
 Since $\sigma$ is generatable by non-uniform $\QPT$ algorithm, adversary $A^\prime$ is also non-uniform $\QPT$. From eq. \ref{eq:indis_kappa}, we get that for infinitely many $\lambda$
  \[
   \Big\vert\Pr[A^\prime(1^\lambda,\rho_0)=1]-\Pr[A^\prime(1^\lambda,\rho_1)=1]\Big\vert=\Big\vert\Pr[A(1^\lambda,\kappa_0)=1]-\Pr[A(1^\lambda,\kappa_1)=1]\Big\vert\geq \frac{1}{q(\lambda)}.
  \]
This contradicts the fact that $\rho_0\approx_{\negl_C}\rho_1$. Hence, the desired statement follows.   
\end{proof}

\begin{fact} \label{fact:datanegl}
Let $\rho_0 \approx_{\negl_C} \rho_1$ and $B$ be a non-uniform $\QPT$ algorithm. Then $B(\rho_0) \approx_{\negl_C} B(\rho_1)$.
\end{fact}
\begin{proof}
    Assume for contradiction that there exists a non-uniform $\QPT$ algorithm $A$ and a positive polynomial $p(\cdot)$ such that for infinitely many $\lambda$,
    \begin{equation*}
        \label{eq:Adisnt}
   \Big\vert\Pr[A(1^\lambda,B(\rho_0))=1]-\Pr[A(1^\lambda, B(\rho_1))=1]\Big\vert \geq \frac{1}{p(\lambda)}.
    \end{equation*}
  We now create a non-uniform adversary $A^\prime$ to distinguish between $\rho_0$ and $\rho_1$ as follows:
  \begin{itemize}
  \item $A^\prime$  takes as input the security parameter $\lambda$ and state $\rho_b$ (for $b \in \{0,1\}$).
      \item $A^\prime$ returns $A(1^\lambda, B(\rho_b))$.
  \end{itemize} 
  Consider,
  \begin{align*}
  \Big\vert\Pr[A^\prime(1^\lambda,\rho_0)=1]-\Pr[A^\prime(1^\lambda, \rho_1)=1]\Big\vert  &=         \Big\vert\Pr[A(1^\lambda,B(\rho_0))=1]-\Pr[A(1^\lambda, B(\rho_1))=1]\Big\vert \\
&\geq \frac{1}{p(\lambda)}.
  \end{align*}
  This is a contradiction. 
\end{proof}

\section{ An $\sOWSG$ implies an $\EFI$}\label{sec:OWSGtoEFI}

\begin{theorem}
    \label{theorem:EFI}
Let $m = \frac{cn}{\log(n)}$ for some constant $c>0$ where $n$ is the key-length of the $\sOWSG$. An $m$-copy $\sOWSG$ (see Definition~\ref{def:OWSG}) implies a $k^*$-imbalanced $\EFI$  (see Definition~\ref{def:imEFI}) for some $k^*(\lambda) \in \poly(\lambda)$. 
\end{theorem}
\begin{proof}
Let $\ext$ be the extractor from Fact~\ref{fact:hashing} and $\ext_Q$ be the quantum extractor from Fact~\ref{fact:quantum_extractor}. Let $i^*$ be from Claim~\ref{claim:D_close} and the total number of qubits in $\tau^{XQ^{i^*}}$ be at most $n^{c''}$ for some constant $c'' > 0$. Let $g(\cdot)$  be the quantum hardcore function from Fact \ref{fact:quantum_hardcore} with output size  $c_1\log(n)$  where $c_1=14c''+3c^{-1}+100$.  Let $l_{i^*}$ be from Claim \ref{claim:new_hardcore} and $\sigma_0(i,l_{i})$ and $\sigma_1(i,l_{i})$ be from Claim~\ref{claim:final_entropy_diff}.

\begin{figure}[h]
    \centering
\fbox{\begin{minipage}{40em} \textbf{\underline{EFI$_k(1^\lambda,b)$}:}
\begin{enumerate}
\item The algorithm takes as input the security parameter $\lambda$ and a bit $b$. 
\item It has access to an $m$-copy $\sOWSG$ with key-length $n$ that generates the one-way state $\tau^{XQ^m}$.
\item   For all $i\in[m],l\in [n]$ define, 
\begin{equation} \label{eq:tau_0_def}
    \tau_0(i,l)\defeq Q^{i}HH^l(X)Rg(X,R)_\tau.
\end{equation}
\begin{equation*}
\rho_0(\tau) \defeq\sum_{i=1}^m\sum_{l=1}^n \frac{1}{mn}\ketbra{i,l}{i,l}\otimes {\tau_0(i,l)}_{}.
\end{equation*} 
\item Let $n_0\in \poly(\lambda)$ be the number of qubits in $\rho_0(\tau)$ and 
\[t_0=\max(n,\lambda), ~t=t_0^{47+2c''}, ~s_Q(k)=4n_0t-k+O(\log(t_0)).\]
\item If $b=0$, output $\ext_Q(\rho_0(\tau)^{\otimes t},U_{s_Q(k)}).$ \item If $b=1$, output $U_{4n_0t+1}$. 
\end{enumerate}
\end{minipage}}
    \caption{$k^*$-imbalanced $\EFI$ }
    \label{fig:EFI}
\end{figure}
\noindent The $k^*$-imbalanced $\EFI$  is given in Figure~\ref{fig:EFI} where $k^*$ is defined as follows.  Define,
\[\rho_1(\tau) \defeq \sum_{j=1}^m\sum_{l=1}^n
 \mathbb{1}(j\neq i^* \text{ or }l\neq l_{i^*}) \cdot \frac{1}{mn}\ketbra{j,l}{j,l}\otimes\sigma_0(j,l)+\frac{1}{mn}\ketbra{i^*,l_{i^*}}{{i^*},l_{i^*}}\otimes\sigma_1(i^*,l_{i^*}),\]
 \[t_0\defeq\max(n,\lambda), ~t\defeq t_0^{47+2c''},~k^*(\lambda) \defeq S_\infty^\eps({\rho}_1(\tau)^{\otimes t}).\]
  From the construction in eq. \ref{eq:def_sigma0}, for all $i$ and $l$, we have $\sigma_0(i,l)=\tau_0(i,l)$.
 This gives
\[
\rho_0(\tau) =\sum_{j=1}^m\sum_{l=1}^n
 \frac{1}{mn}\ketbra{j,l}{j,l}\otimes\sigma_0(j,l).
\]
From Claim \ref{claim:final_entropy_diff}, since $\sigma_0(i^*,l_{i^*})\approx_{\negl_C}\sigma_1(i^*,l_{i^*})$, {we get $\rho_0(\tau)\approx_{\negl_C}\rho_1(\tau)$ (Fact \ref{fact:indistinguishability_convex_combination}).} Consider,
\begin{align}
\label{eq:before_repetition_extractor}
S(\rho_1(\tau))-S(\rho_0(\tau)) & = \frac{1}{mn}\left[S(\sigma_1(i^*,l_{i^*}))-S(\sigma_0(i^*,l_{i^*}))\right] & \mbox{(Fact~\ref{fact:chainentropy})} \nonumber\\
&\geq \frac{48\log(n)-4}{n^{22+c''}}. & \mbox{(Claim~\ref{claim:final_entropy_diff})}
\end{align}
 Using Fact \ref{fact:smooth_entopy_multiple_copies}  and Fact \ref{fact:alphabit} and setting $\eps=2^{-\log^2(\lambda)}$, we get   
\[S_0^\eps(\rho_0(\tau)^{\otimes t})\leq t\cdot S(\rho_0(\tau))+O(\sqrt{tn})\log(\lambda),\] 
\[S_\infty^\eps({\rho}_1(\tau)^{\otimes t})\geq t\cdot S({\rho}_1(\tau))-O(\sqrt{tn})\log(\lambda).\] 
Since $t_0=\max(n,\lambda)$ and $t=t_0^{47+2c''}$, we get \begin{equation}
    \label{eq:support_difference}
S_\infty^\eps({\rho}_1(\tau)^{\otimes t})-S_0^\eps(\rho_0(\tau)^{\otimes t})\geq t_0^{25+c''}\cdot {(48\log(n)-4)}-O(t_0^{24+c''}\log(t_0))\geq t_0^{25+c''}.\end{equation}
 Note that \begin{align*}
    k^*(\lambda)&=  S_\infty^\eps({\rho}_1(\tau)^{\otimes t}) \\
    &\leq n_0t \in \poly(\lambda). &\mbox{(Fact~\ref{fact:alpha_renyi_inequalities})} 
    \end{align*} 
 Define,
\[\widetilde{\sigma}_0(k,\tau) \defeq \ext_Q(\rho_0(\tau)^{\otimes t},U_{s_Q(k)}),
\]
\[\widetilde{\sigma}_1(k,\tau) \defeq \ext_Q(\rho_1(\tau)^{\otimes t},U_{s_Q(k)}).
\]
\begin{enumerate}
    \item  When $k\leq k^*(\lambda)$:

From Fact \ref{fact:quantum_extractor}, we see that if $k\leq k^*(\lambda)$, we get
\begin{equation}  \label{eq:def_sigma_1}\widetilde{\sigma}_1(k,\tau)=\ext_Q({\rho}_1(\tau)^{\otimes t},U_{s_Q(k)})\approx_{\negl(\lambda)}U_{4n_0t+1}.
\end{equation}
Since $\rho_0(\tau)\approx_{\negl_C}\rho_1(\tau)$, Fact \ref{fact:indistinguishability_multiple_samples}, Fact~\ref{fact:datanegl} and eq. \ref{eq:def_sigma_1}  gives us $\widetilde{\sigma}_0(k,\tau)\approx_{\negl_C}U_{4n_0t+1}$ for all $k\leq k^*(\lambda)$. 

  \item  When $k\geq k^*(\lambda)$:
  
From eq. \ref{eq:support_difference},  $S_0^\eps(\rho_0(\tau)^{\otimes t})\leq k^*(\lambda)-t_0^{25+c''}$. Using Fact \ref{fact:quantum_extractor}, we get \[ S_0^\eps(\widetilde{\sigma}_0(k,\tau))\leq k^*(\lambda)+ 4n_0t-k+O(\log(t_0))-t_0^{25+c''}\leq  4n_0t-t_0^{25+c''}+O(\log(t_0)).\]

{This means there exists a state $\widetilde{\sigma}^\prime \in \mathcal{B}^\eps(\widetilde{\sigma}_0(k,\tau))$ such that $|\support(\widetilde{\sigma}^\prime)| \leq 2^{4n_0t-t_0^{25+c''}+O(\log(t_0))}$, which is much smaller than the support-size of $U_{4n_0t+1}$. 

Thus, we get \[\onenorm{\widetilde{\sigma}^\prime}{U_{4n_0t+1}}\geq 1-\frac{2^{4n_0t-t_0^{25+c''}+O(\log(t_0))}}{2^{4n_0t+1}}\geq 
1-\negl(\lambda).\] 

Since $\eps=\negl(\lambda)$ and $\Delta_B(\widetilde{\sigma}^\prime,\widetilde{\sigma}_0(k,\tau))\leq \eps$, using Fact \ref{fact:fuchs} and triangle inequality, we get
\[\onenorm{\widetilde{\sigma}_0(k,\tau)}{U_{4n_0t+1}}\geq 
1-\negl(\lambda).\] }
\end{enumerate}
We see that for $k\leq k^*(\lambda)$, we get computational indistinguishability of $\widetilde{\sigma}_0(k,\tau)$ and $U_{4n_0t+1}$ and  for $k\geq k^*(\lambda)$, we get statistical difference between  $\widetilde{\sigma}_0(k,\tau)$ and $U_{4n_0t+1}$. Thus, we get a $k^*$-imbalanced $\EFI$ where EFI$_k(1^\lambda, 0)=\widetilde{\sigma}_0(k,\tau)$ and EFI$_k(1^\lambda, 1)=U_{4n_0t+1}$.
\end{proof}
\begin{claim}\label{claim:D_close}
    There exists  $i^*\in [m]$ such that \[
\mathrm{D}(\tau^{XQ^{i^*+1}}\|\Id^X\otimes\tau^{Q^{i^*+1}})-\mathrm{D}(\tau^{XQ^{i^*}}\|\Id^X\otimes\tau^{Q^{i^*}})\leq c^{-1}\log(n). \]
\end{claim}
\begin{proof}
     Consider,
    \begin{align*}
        &\sum_{i=0}^{m-1}\left(\mathrm{D}(\tau^{XQ^{i+1}}\|U^X\otimes\tau^{Q^{i+1}})-\mathrm{D}(\tau^{XQ^{i}}\|U^X\otimes\tau^{Q^{i}})\right)\\&=\mathrm{D}(\tau^{XQ^{m}}\|U^X\otimes\tau^{Q^{m}})-\mathrm{D}(\tau^{X}\|U^X)
         \\&\leq \mathrm{D}(\tau^{XQ^{m}}\|U^X\otimes\tau^{Q^{m}}) 
         &\mbox{(Fact \ref{fact:non_negative})} 
         \\&\leq \dmax{\tau^{XQ^{m}}}{U^X\otimes\tau^{Q^{m}}}  &\mbox{(Fact \ref{fact:monotonicity_renyi})} 
       \\&\leq n.  &\mbox{(Fact~\ref{fact:identity_upper})}
    \end{align*}
     For all $i$, using Fact \ref{fact:data}, we also have\[
\mathrm{D}(\tau^{XQ^{i+1}}\|U^X\otimes\tau^{Q^{i+1}})-\mathrm{D}(\tau^{XQ^{i}}\|U^X\otimes\tau^{Q^{i}})\geq 0.
   \]
    Since $m=cn/\log(n)$,  there exists  $i^*\in [m]$ such that \begin{align*} 
\mathrm{D}(\tau^{XQ^{i^*+1}}\|\Id^X\otimes\tau^{Q^{i^*+1}})-\mathrm{D}(\tau^{XQ^{i^*}}\|\Id^X\otimes\tau^{Q^{i^*}})&=\mathrm{D}(\tau^{XQ^{i^*+1}}\|U^X\otimes\tau^{Q^{i^*+1}})-\mathrm{D}(\tau^{XQ^{i^*}}\|U^X\otimes\tau^{Q^{i^*}})\\&\leq c^{-1}\log(n). \qedhere
     \end{align*}
\end{proof}
\noindent In Claim \ref{claim:D_close}, we identified an $i^*\in [m]$ where\[
\mathrm{D}(\tau^{XQ^{i^*+1}}\|\Id^X\otimes\tau^{Q^{i^*+1}})\approx\mathrm{D}(\tau^{XQ^{i^*}}\|\Id^X\otimes\tau^{Q^{i^*}}).
\] 
 Here $\approx$ means the quantities have a difference at most $O(\log(n))$. In order to get the imbalanced-$\EFI$, we want $
 S(X|Q^{i^*}HH^l(X))_\tau\approx 0,$
which requires $l\approx S(X|Q^{i^*})_\tau$. However, due to the properties of the extractor, we can only extract up to $l= S_2(X|Q^{i^*})_\tau$. Therefore, we want $l\approx S(X|Q^{i^*})_\tau\approx S_2(X|Q^{i^*})_\tau$. From the properties of the quantum hardcore function, we can argue that \[
Q^{i^*}HH^l(X)Rg(X,R)_{\tau}\approx_{\negl_C}Q^{i^*}HH^l(X)_{\tau}\otimes U_{|R|} \otimes U_{|g(X,R)|},
    \] 
    if we can ensure that $XQ^{i^*}HH^l(X)_{\tau}$ is a one-way state (see Definition \ref{def:OWSG}).  For this to hold, we require \[l\approx S_2(X|Q^{i^*+1})_{\tau}\approx S_2(X|Q^{i^*})_{\tau}.\] 
    In Claim \ref{claim:combined}, we identify a substate $\gamma$ of ${\tau}$ (see Definition \ref{def:substate}) with weight at least $\frac{1}{\poly(n)}$ which satisfies the above requirements. In other words, $\gamma$ satisfies
\begin{align*}
         l_{i^*}\approx S_2(X|Q^{i^*+1})_{\gamma}\approx S_2(X|Q^{i^*})_{\gamma}\approx S(X|Q^{i^*})_{\gamma}.
         \end{align*}
This claim is a key technical contribution of this work. To enhance readability, we provide a proof sketch here and defer the full proof to the appendix. 
\begin{claim}\label{claim:combined}
    There exists a state $\gamma$ such that
    \begin{enumerate}[(1)]
        \item $\dmax{\gamma}{\tau}\leq (20+c'')\log(n)$. \label{itm1}
        \item $
S(X|Q^{i^*})_{\gamma}-S_2(X|Q^{i^*+1})_{\gamma} \leq (30+3c^{-1}+11c'')\log(n)+1.$ \label{itm2}
    \end{enumerate}
 \end{claim}
 \begin{proof}(Sketch) Let $i^*$ be from Claim~\ref{claim:D_close}. We start by noting that for ``flat'' states (that is all non-zero eigenvalues are the same) all R\'enyi entropies are the same. This may help us relate $S_2(\cdot)$ and $S(\cdot)$. Since we need to do this for conditional entropies we define \begin{equation*}
\theta^{XQ^m}\defeq(\tau^{Q^{i^*}})^{\frac{-1}{2}}\tau^{XQ^m}(\tau^{Q^{i^*}})^{\frac{-1}{2}}.
\end{equation*} 
We write $\theta$ as
\[
\theta^{XQ^m}=\sum_{j\in \poly(n)}\theta_j^{XQ^m},
\] where for all $j\neq0$, all eigenvalues of $\theta_j$ are within a factor of $2$ of each other, i.e., $\theta_j$ is ``2-flat". Note that the total number of $j$'s is at most a polynomial in $n$. We get 
\begin{equation*}
{\tau}^{XQ^{m}}=(\tau^{Q^{i^*}})^{\frac{1}{2}}{\theta}^{XQ^{m}}(\tau^{Q^{i^*}})^{\frac{1}{2}}=\sum_j(\tau^{Q^{i^*}})^{\frac{1}{2}}\theta_j^{XQ^m}(\tau^{Q^{i^*}})^{\frac{1}{2}}.
\end{equation*}
We extend ${\tau}$ as \begin{equation*}
{\tau}^{JXQ^{m}}\defeq\sum_j p_j\ketbra{j}{j}^J\otimes{\tau}_j^{XQ^{m}}, 
\end{equation*}
where \[p_j\defeq\Tr((\tau^{Q^{i^*}})^{\frac{1}{2}}\theta_j^{XQ^m}(\tau^{Q^{i^*}})^{\frac{1}{2}}), \]\[\sum_jp_j=1,\]\begin{equation}
\tau_j^{XQ^{m}}\defeq\frac{(\tau^{Q^{i^*}})^{\frac{1}{2}}\theta_j^{XQ^m}(\tau^{Q^{i^*}})^{\frac{1}{2}}}{\Tr((\tau^{Q^{i^*}})^{\frac{1}{2}}\theta_j^{XQ^m}(\tau^{Q^{i^*}})^{\frac{1}{2}})}.\label{eq:tau-j}\end{equation}     
The above process helps us with $S_2(X|Q^{i^*})_{\tau_j}\approx S(X|Q^{i^*})_{\tau_j}$. However, we furthermore need 
$S_2(X|Q^{i^*+1})_\gamma\approx S(X|Q^{i^*})_\gamma$. For this, we consider the eigendecomposition 
\[(\tau^{Q^{i^*+1}})^{\frac{-1}{2}}{\tau}_j^{XQ^{i^*+1}}(\tau^{Q^{i^*+1}})^{\frac{-1}{2}}=\sum_i\lambda_{i,j}\ketbra{e_{i,j}}{e_{i,j}}^{XQ^{i^*+1}}.\] 
We consider two parts, one with large eigenvalues (crossing a threshold) and the residual part. We define 
\[\good_j \defeq\{i: \lambda_{i,j}<2^{t_{j}}, ~t_{j}:=\diver{{\tau}_j^{XQ^{i^*}}}{\Id^X\otimes\tau^{Q^{i^*}}}+ O(\log(n))\},\]
 \begin{equation*} 
\Theta_j \defeq\sum_{i\in \good_j}\lambda_{i,j}\ketbra{e_{i,j}}{e_{i,j}}^{XQ^{i^*+1}}, 
\end{equation*}
\[
\bar{\Theta}_j \defeq \sum_{i\notin \good_j}\lambda_{i,j}\ketbra{e_{i,j}}{e_{i,j}}^{XQ^{i^*+1}}.
\] 
We first show that if for all $j$ such that $p_j>\frac{1}{\poly(n)}$, \begin{align*}
weight(\good_j)\defeq\Tr((\tau^{Q^{i^*+1}})^{\frac{1}{2}}\Theta_j^{XQ^{i^*+1}}(\tau^{Q^{i^*+1}})^{\frac{1}{2}})< \frac{1}{\poly(n)}, \end{align*} then we get a substate of $\tau_j$ with weight at least  $1-\frac{1}{\poly(n)}$ for which relative-entropy (with respect to $\Id^X\otimes \tau^{Q^{i^*+1}}$) is large (at least $t_j$). From Fact \ref{fact:quantum_chaining}, we get that $\mathrm{D}(\tau_j^{XQ^{i^*+1}}\|\Id^X\otimes \tau^{Q^{i^*+1}})$ is large. Since this holds for all $j$ such that $p_j>\frac{1}{\poly(n)}$, using Fact \ref{fact:quantum_chaining} again, we are able to argue that $\mathrm{D}(\tau^{XQ^{i^*+1}}\|\Id^X\otimes\tau^{Q^{i^*+1}}) - \mathrm{D}(\tau^{XQ^{i^*}}\|\Id^X\otimes\tau^{Q^{i^*}})$ is large and hence reaching a contradiction to Claim \ref{claim:D_close}. 

Thus, there exists a $j$ (we also argue that it is not $0$) for which $p_j>\frac{1}{\poly(n)}$ and 
\begin{align}
weight(\good_j)=\Tr((\tau^{Q^{i^*+1}})^{\frac{1}{2}}\Theta_j^{XQ^{i^*+1}}(\tau^{Q^{i^*+1}})^{\frac{1}{2}})\geq \frac{1}{\poly(n)}. \label{eq:identigy_j}\end{align}
We conjugate $\Theta_j$ by $(\tau^{Q^{i^*+1}})^\frac{1}{2}$ and normalise to get a substate $\sigma_j$ of $\tau$ with weight at least $\frac{1}{\poly(n)}$. From the definition of $\Theta_j$, we see that all eigenvalues of $\Theta_j$ are small (below the threshold $2^{t_{j}}$). This helps in lower-bounding  $S_2(X|Q^{i^*+1})_{\sigma_j}$ (via upper bounding $\mathrm{D}_2(\cdot)$).
Since $\sigma_j$ is a substate of $\tau_j$, eq. \ref{eq:tau-j} (and that $\theta_j$ is $2$-flat) allows us to lower-bound $\mathrm{D}(\sigma_j^{XQ^{i^*}}\|\Id^X\otimes \tau^{Q^{i^*}})$. Combining this with the observation that $\sigma_j$ is a substate of $\tau$ (with good weight)  allows us to upper-bound $S(X|Q^{i^*})_{\sigma_j}$. Setting $\gamma=\sigma_j$, we get         \[\dmax{\gamma}{\tau}\leq O(\log(n)),\] \[S(X|Q^{i^*})_{\gamma}\approx S_2(X|Q^{i^*+1})_{\gamma}.\]
 \end{proof}
\begin{claim} \label{claim:new_oneway_rhoj*}
   For any non-uniform $\QPT$ adversary $A$,   \[
\Pr[\tang \leftarrow \ver(A(1^\lambda,Q^{i^*}_{\gamma}), {\gamma}^{Q_{i^*+1}} )]=\negl(\lambda).
\]
\end{claim}
\begin{proof}
Let us assume for contradiction that there exists a non-uniform $\QPT$ algorithm $A$ and a positive polynomial $p(\cdot)$, such that for infinitely many $\lambda$,  \begin{equation}
    \label{eq:inversion1}
\delta_\lambda\defeq
\Pr[\tang \leftarrow \ver(A(1^\lambda,Q^{i^*}_{\gamma}), {\gamma}^{Q_{i^*+1}} )]\geq \frac{1}{p(\lambda)}. 
\end{equation}
 From Claim \ref{claim:combined}, we have \begin{equation} \label{eq:inversion2}
    \tau^{XQ^{m}}
=\frac{1}{n^{20+c''}}\gamma^{XQ^{m}}+(1-\frac{1}{n^{20+c''}}) \widetilde{\rho}^{XQ^{m}}. 
\end{equation}
Combining eqs.~\ref{eq:inversion1} and \ref{eq:inversion2}, we get\[
\Pr[\tang \leftarrow \ver(A(1^\lambda,Q^{i^*}_{\tau}), \tau^{Q_{i^*+1}} )]\geq \frac{1}{n^{20+c''}p(\lambda)}. 
\]
This contradicts that $\tau^{XQ^{m}}$ is a one-way state. 
Thus, for any non-uniform $\QPT$ adversary $A$,   \[
\Pr[\tang \leftarrow \ver(A(1^\lambda,Q^{i^*}_{\gamma}), {\gamma}^{Q_{i^*+1}} )]=\negl(\lambda). \qedhere
\]
\end{proof}

\begin{claim}\label{claim:new_hardcore}
  Let $l_{i^*}\defeq S_2({X}|Q^{i^*+1})_{{\gamma}}-(2c''+2)\log(n)$. 
There exists a function $g:\{0,1\}^n\times \{0,1\}^{2n} \rightarrow \{0,1\}^{c_1\log(n)}$ such that 
  \[\tau^\prime_0({i^*},l_{i^*})\approx_{\negl_C}\tau^\prime_1({i^*},l_{i^*}),\]

  \noindent where $c_1=14c''+3c^{-1}+100$ and
\[\tau^\prime_0({i},l)\defeq Q^{i}HH^l({X}) R g(X,R)_{\gamma},\] 
\[\tau^\prime_1({i},l)\defeq Q^{i}HH^l({X})_{\gamma} \otimes U_{2n}   \otimes U_{c_1\log(n)}.\]
\end{claim}
\begin{proof}
 Let us assume for contradiction that there exists a non-uniform $\QPT$ algorithm $A$ and a positive polynomial $p(\cdot)$, such that for infinitely many $\lambda$,  \begin{equation}
    \label{eq:invert_1}
\delta_\lambda\defeq
\Pr[\tang \leftarrow \ver(A(1^\lambda,{\gamma}^{Q^{i^*}HH^{l_{i^*}}(X)}), {\gamma}^{Q_{i^*+1}} )]\geq \frac{1}{p(\lambda)}. 
\end{equation}
Let $l^\prime= l_{i^*} -4\log(\frac{4}{\delta_\lambda})$.  From Fact \ref{fact:hashing}, we get
\begin{equation}
         \label{eq:ext_closeness_1}
     Q^{i^*+1}HH^{l^\prime}(X)_{\gamma}\approx_{{\delta_\lambda^2/16}} Q^{i^*+1}_{\gamma}\otimes U_s\otimes U_{l^\prime}.
     \end{equation}
From Fact~\ref{fact:extension}, using \[\rho_A \leftarrow Q^{i^*+1}HH^{l^\prime}(X)_{\gamma}~,~ \rho_{AB} \leftarrow Q^{i^*+1}HH^{l_{i^*}}(X)_{\gamma}~,~ \sigma_A \leftarrow Q^{i^*+1}_{\gamma}\otimes U_s\otimes U_{l^\prime},\] we get an extension $\theta_{AB}$ of $\sigma_A$ such that 
\begin{align*}
    \onenorm{\rho_{AB}}{\theta_{AB}}
    &\leq 2 \sqrt{2} \cdot \Delta_B (\rho_{AB}, \theta_{AB}) & \mbox{(Fact~\ref{fact:fuchs})}\\
    & = 2 \sqrt{2} \cdot \Delta_B(\rho_A, \sigma_A) \\
    &\leq \frac{\delta_\lambda}{2}. & \mbox{(eq. \ref{eq:ext_closeness_1} and Fact~\ref{fact:fuchs})}\\
\end{align*}
Hence,
\[
\Pr[\tang \leftarrow \ver(A(1^\lambda,\theta^{Q^{i^*}HH^{l_{i^*}}(X)}), \theta^{Q_{i^*+1}} )] \geq \delta_\lambda-\frac{\delta_\lambda}{{2}}\geq \frac{\delta_\lambda}{2}.
\]
From Fact \ref{fact:identity_upper}, we have
\begin{align*}
    &\theta \leq 2^{8\log\left(\frac{4}{\delta_\lambda}\right)}Q^{i^*+1}_{\gamma}\otimes U_s\otimes U_{l_{i^*}}.
\end{align*}
Thus, 
\begin{equation}
    \label{eq:invert_2}
\Pr[\tang \leftarrow \ver(A(1^\lambda,{Q^{i^*}_{\gamma}\otimes U_s\otimes U_{l_{i^*}}}), {\gamma}^{Q_{i^*+1}} )] \geq \frac{\delta_\lambda}{2}\cdot \left(\frac{\delta_\lambda}{4}\right)^8.
\end{equation}
We now construct a non-uniform adversary $A^\prime$ as follows. \\\\
\textbf{$A^\prime(1^\lambda, Q^{i^*}_{\gamma})$:}
\begin{enumerate}
\item $A^\prime$ takes as input   $Q^{i^*}_{\gamma}$ along with the security parameter $\lambda$.
\item $A^\prime$ takes as (non-uniform) advice $l_{i^*}$.
 \item $A^\prime$ appends $U_s\otimes U_{l_{i^*}}$ to $Q^{i^*}_{\gamma}$. 
 \item $A^\prime$ returns $A(1^\lambda, Q^{{i^*}}_{\gamma}\otimes U_s\otimes U_{l_{i^*}} )$. 
\end{enumerate}
From eqs. \ref{eq:invert_1} and \ref{eq:invert_2}, we see that   for infinitely many $\lambda$,
\begin{equation*}
\Pr[\tang \leftarrow \ver(A^\prime(1^\lambda,{\gamma}^{Q^{i^*}}), {\gamma}^{Q_{i^*+1}} )]\geq \frac{1}{2^{17} \cdot p(\lambda)^9}. 
\end{equation*}
This contradicts Claim \ref{claim:new_oneway_rhoj*}. Thus, for any non-uniform $\QPT$ adversary $A$,   \[
\Pr[\tang \leftarrow \ver(A(1^\lambda,{\gamma}^{Q^{i^*}HH^{l_{i^*}}(X)}), {\gamma}^{Q_{i^*+1}} )]= \negl(\lambda).
\]
 The claim now follows by taking $g(\cdot)$ as the quantum hardcore function from Fact \ref{fact:quantum_hardcore}.
\end{proof}

\begin{claim}\label{claim:new_entropy_diff_large}
\[S(\tau^\prime_1(i^*,l_{i^*} ))-S(\tau^\prime_0(i^*,l_{i^*}))\geq (68+c'')\log(n)-2.\]
 \end{claim}
\begin{proof}
{Since $l_{i^*}=S_2({X}|Q^{{i^*}+1})_{\gamma}-(2c''+2) \log(n)$}, Fact \ref{fact:hashing} gives
     \begin{equation}\label{eq1}
     Q^{i^*+1}HH^{l_{i^*}}({X})_{\gamma}\approx_{2^{-{(c''+1)}\log(n)}} Q^{i^*+1}_{\gamma}\otimes U_s\otimes U_{l_{i^*}}.\end{equation}
        On tracing out $Q_{i^*+1}$, we have
    \begin{equation}
        \label{eq*:hash_traced_out}
     Q^{i^*}HH^{l_{i^*}}(X)_{\gamma}\approx_{2^{-(c''+1)\log(n)}} Q^{i^*}_{\gamma}\otimes U_s\otimes U_{l_{i^*}}.\end{equation}
     Let 
     \[\widetilde{\tau}_1({i^*},l_{i^*})=Q^{i^*}_{\gamma}\otimes U_{s} \otimes U_{l_{i^*}} \otimes U_{2n}\otimes U_{c_1\log(n)}.\] 
\begin{align*}
        &\onenorm{\tau^\prime_1({i^*},l_{i^*})}{\widetilde{\tau}_1({i^*},l_{i^*})}\\
        &=\onenorm{Q^{i^*}HH^{l_{i^*}}(X)_{\gamma} \otimes U_{2n}\otimes U_{c_1\log(n)}}{Q^{i^*}_{\gamma}\otimes U_{s} \otimes U_{l_{i^*}} \otimes U_{2n}\otimes U_{c_1\log(n)}}\\&=\onenorm{Q^{i^*}HH^{l_{i^*}}(X)_{\gamma} }{Q^{i^*}_{\gamma}\otimes U_{s} \otimes U_{l_{i^*}} }\\&\leq \frac{1}{n^{c''+1}}. \tab \tab \tab \tab \tab \mbox{(eq. \ref{eq*:hash_traced_out})}       \end{align*}
From Fact \ref{fact:fannes}, we get (here $\log(d)$ is the number of qubits in $\tau^\prime_1({i^*},l_{i^*})$)
\begin{equation}
           \label{eq*:fannes}
\vert S(\tau^\prime_1({i^*},l_{i^*})) -S(\widetilde{\tau}_1({i^*},l_{i^*}))\vert \leq \log(d)\onenorm{\tau^\prime_1({i^*},l_{i^*})}{\widetilde{\tau}_1({i^*},l_{i^*})}+\frac{1}{e} \leq \frac{O(n^{c''})}{n^{c''+1}}+\frac{1}{e}\leq 1.
 \end{equation}
Consider,
 \begin{align}
&S(\tau^\prime_0({i^*},l_{i^*})) \nonumber \\
&=S(Q^{i^*}HH^{l_{i^*}}({X}) R g(X,R))_{\gamma} \nonumber
     \\
     &\leq S({X}Q^{i^*}HH^{l_{i^*}}({X}) R g(X,R))_{\gamma}  \nonumber&\mbox{(Fact \ref{fact:entropy_inequalities})}\\&= S({X}Q^{i^*}HR)_{\gamma}  \nonumber&\mbox{(since $h^{l_{i^*}}({x}),g(x,r)$ are functions of ${x},h,r$)}
     \\&= S({X}Q^{i^*}\otimes U_s\otimes U_{2n})_{\gamma}\nonumber
     \\
     &= S({X}Q^{i^*})_{\gamma}+s+2n  \nonumber\\
      &= S({X}|Q^{i^*})_{\gamma}+S(Q^{i^*})_{\gamma}+s+2n  &\mbox{(Fact \ref{fact:chainentropy})}\nonumber\\
      &\leq   S_2(X|Q^{i^*+1})_{\gamma} +S(Q^{i^*})_{\gamma}+s+2n \nonumber\\&\quad+(30+3c^{-1}+11c'')\log(n)+1. &\mbox{(Claim \ref{claim:combined})}\label{eq:rho_0_entropy}
 \end{align}
Consider,
\begin{align}
&S(\widetilde{\tau}_1({i^*},l_{i^*})) \nonumber \\
&=S(Q^{i^*})_{\gamma}+s+l_{i^*}+2n+c_1\log(n)\nonumber\\&= S(Q^{i^*})_{\gamma}+s+S_2({X}|Q^{i^*+1})_{{\gamma}}-(2c''+2)\log(n)+2n+c_1\log(n).\label{eq:rho_1_entropy}
\end{align}
Combining eqs. \ref{eq*:fannes}, \ref{eq:rho_0_entropy} and \ref{eq:rho_1_entropy}, using $c_1=14c''+3c^{-1}+100$, we get  
\[S({\tau^\prime_1}(i^*,l_{i^*}))-S(\tau^\prime_0(i^*,l_{i^*}))\geq (68+c'')\log(n)-2. \qedhere\] 

\end{proof}

\begin{claim}\label{claim:final_entropy_diff}
\noindent Let $\widetilde{\rho}$ be according to eq. \ref{eq:inversion2}. Define \[\tau^\prime_2({i},l)\defeq Q^{i}HH^l({X}) R g(X,R)_{\widetilde{\rho}}.\]
    \begin{equation} \label{eq:def_sigma0}
         \sigma_0(i,l)\defeq\frac{1}{n^{20+c''}}\tau^\prime_0({i},l)+(1-\frac{1}{n^{20+c''}})\tau^\prime_2(i,l).
    \end{equation}
     \begin{equation} \label{eq:def_sigma1}
    \sigma_1(i,l)=\frac{1}{n^{20+c''}}\tau^\prime_1({i},l)+(1-\frac{1}{n^{20+c''}})\tau^\prime_2(i,l).
    \end{equation}
    Then,
\begin{equation*}
\sigma_0(i^*,l_{i^*})\approx_{\negl_C}\sigma_1(i^*,l_{i^*}).   \end{equation*}
    \[ S(\sigma_{1}(i^*,l_{i^*}))-S(\sigma_{0}(i^*,l_{i^*}))\geq \frac{48\log(n)-4}{n^{20+c''}} .\]
\end{claim}
\begin{proof}
  From eqs. \ref{eq:def_sigma0} and \ref{eq:def_sigma1}, using Claim \ref{claim:new_hardcore} and Fact \ref{fact:indistinguishability_convex_combination}, we get \begin{equation}
    \label{eq:indis}\sigma_0(i^*,l_{i^*})\approx_{\negl_C}\sigma_1(i^*,l_{i^*}).   \end{equation}
 From eqs. \ref{eq:def_sigma0} and \ref{eq:def_sigma1}, using Fact \ref{fact:convex_entropy_diff} with $p=\frac{1}{n^{20+c''}}$,  we get
  \begin{align*}
      S(\sigma_{1}(i^*,l_{i^*}))-S(\sigma_{0}(i^*,l_{i^*}))&\geq \frac{1}{n^{20+c''}} \cdot (S(\tau_1^\prime(i^*,l_{i^*}))-S(\tau_0^\prime(i^*,l_{i^*}))-\log(n^{20+c''})-2)
      \\&=\frac{1}{n^{20+c''}} \cdot ((68+c'')\log(n)-2-(20+c'')\log(n)-2) &\mbox{(Claim \ref{claim:new_entropy_diff_large})}
       \\&=\frac{48\log(n)-4}{n^{20+c''}}.  \qedhere
  \end{align*} 

\end{proof}

\section{An $\EFI$ implies an $\sOWSG$}\label{sec:EFItoOWSG}
\begin{theorem}
    An $\EFI$  implies a $\poly(n)$-copy $\sOWSG$.
\end{theorem}
\begin{proof}
    
Suppose we have an $\EFI$ pair $(\rho_0,\rho_1)$. Then, $\rho_0\approx_{\negl_C}\rho_1$.
From Fact \ref{fact:EFI_amplification}, we can also assume without loss of generality that $\frac{1}{2}\onenorm{\rho_0}{\rho_1}\geq 1-\negl(\lambda)$.

We create an $\sOWSG$ as follows:
\begin{enumerate}
\item $\keygen(1^\lambda)\rightarrow x$ : $x\leftarrow U_n$ for $n=\lambda$. 
    \item $\gen(1^\lambda, x)\rightarrow \phi_x: \phi_x=\rho_{x_1}\otimes \rho_{x_2}\dots\rho_{x_n} $ where $x_i$ represents the $i^{th}$-bit of $x$.
    \item $\ver(x^\prime, \phi_x)\rightarrow \tang /\bot:$
      Let $\{\pi_0, \pi_1\}$ be the optimal distinguisher for $\rho_0$ and $\rho_1$. $\ver$ measures  $\phi_x$ according to the projectors   $\{\pi_{x^\prime_1}\otimes \pi_{x^\prime_2} \dots \pi_{x^\prime_n},\Id-\pi_{x^\prime_1}\otimes \pi_{x^\prime_2} \dots \pi_{x^\prime_n}\} $. It outputs $\tang$ if the first result is obtained and outputs $\perp$  otherwise.
\end{enumerate}
We know that since $\pi_0,\pi_1$ are the optimal distinguishers and $\frac{1}{2}\onenorm{\rho_0}{\rho_1}\geq 1-\negl(\lambda)$, we have the following: 
\[\Tr(\pi_0\rho_0)\geq 1-\negl(\lambda) ~,~ \Tr(\pi_1\rho_1)\geq 1-\negl(\lambda),\]
 \[\Tr(\pi_1\rho_0)= \negl(\lambda) ~,~ \Tr(\pi_0\rho_1)= \negl(\lambda).\]
\subsubsection*{Correctness:}
\noindent From the union bound, we get  
\begin{equation}
    \label{eq:ver_1}
\Pr(\tang\leftarrow\ver(x,\phi_x))\geq 1-n\cdot\negl(\lambda)\geq 1-\negl(\lambda). 
\end{equation} 
The above gives us the correctness condition required for an $\sOWSG$.
\subsubsection*{Security:}
\noindent We see that for any $x^\prime\neq x$,
\begin{equation}
\label{eq:ver_2}\Pr(\tang\leftarrow\ver(x^\prime,\phi_x))= \negl(\lambda). \end{equation}
Let $X\leftarrow U_{n}$. Assume for contradiction that there exists a non-uniform $\QPT$ adversary $A$ and $t\in \poly(n)$ such that  
\begin{align*}
    \Pr(\tang \leftarrow \ver(A(1^\lambda,\phi_X^{\otimes t}),\phi_X))\geq \frac{1}{q(n)},
\end{align*}
for some positive polynomial $q(\cdot)$. Let $A(1^\lambda,\rho_{X_1}^{\otimes t}\otimes \rho_{X_2}^{\otimes t}\dots\rho_{X_n}^{\otimes t})=X^\prime$. Combining the above with eqs. \ref{eq:ver_1} and \ref{eq:ver_2} gives us 
\begin{equation}
    \Pr_{}(X'=X) = \Pr_{}(A(1^\lambda,\phi_X^{\otimes t})=X)\geq \frac{1}{2q(n)}. \label{eq:highprob}
\end{equation}
 We also have\[
\Pr(X^\prime=X)=\prod_{i=1}^{n}\Pr(X_i^\prime=X_i|X^\prime_{[i-1]}=X_{[i-1]}).
\]
Since $\Pr(X^\prime=X)\geq \frac{1}{2q(n)}$, there exists an $i\in [n]$ such that \begin{equation}
    \label{eq:nonuniform_i}
\Pr(X_i^\prime=X_i|X^\prime_{[i-1]}=X_{[i-1]})\geq \frac{3}{4}.\end{equation}
We now construct a non-uniform adversary $A^\prime$ as described below. Let $B \leftarrow U_1$. \\\\
\textbf{$A^\prime(1^\lambda, \rho_B^{\otimes t})$:}
\begin{enumerate}
\item $A^\prime$ takes as input $t$-copies of the state $\rho_B$ along with the security parameter $\lambda$. It also has access to an $\EFI$  that generates the $\EFI$ pair $(\rho_0,\rho_1)$.
\item $A^\prime$ takes as advice $i$ from  eq.~\ref{eq:nonuniform_i}.
\item $A^\prime$ generates $X_1 \dots X_{i-1} X_{i+1} \dots X_n$ each drawn i.i.d from $U_1$.
\item $A^\prime$ sets $X^\prime=A(1^\lambda,\rho_{X_1}^{\otimes t} \otimes \dots \rho_{X_{i-1}}^{\otimes t} \otimes \rho_B^{\otimes t} \otimes \rho_{X_{i+1}}^{\otimes t} \dots \rho_{X_n}^{\otimes t})$.
\item  $A^\prime$ checks if $X_{[i-1]}^\prime=X_{[i-1]}$. \begin{itemize}
    \item If $X_{[i-1]}^\prime=X_{[i-1]}$, then $A^\prime$ outputs $X^\prime_i$.
    \item Otherwise, $A^\prime$ outputs a random bit $U_1$.
\end{itemize}
\end{enumerate}
Define $X = X_1 \dots X_{i-1} B X_{i+1} \dots X_n$. Note that $X \leftarrow U_n$. From the above construction, we get
\begin{align*}
    &\Pr(A^{\prime}(1^\lambda, \rho_B^{\otimes t})=B)\\
    &=\Pr(X^\prime_{[i-1]}=X_{[i-1]}).\Pr(X^\prime_i=B|X^\prime_{[i-1]}=X_{[i-1]}) + \Pr(X^\prime_{[i-1]} \neq X_{[i-1]}). \frac{1}{2}.
\end{align*}
From eq.~\ref{eq:highprob} we have $\Pr(X^\prime=X)\geq \frac{1}{2q(n)}$ which gives $\Pr(X^\prime_{[i-1]}=X_{[i-1]})\geq \frac{1}{2q(n)}$. 
Combining this with eq. \ref{eq:nonuniform_i}, we get
\begin{align*}
    &\Pr(A^{\prime}(1^\lambda, \rho_B^{\otimes t})=B)\geq \frac{1}{2}+\frac{1}{8q(n)},
    \\ \Rightarrow~ & \frac{1}{2}\Pr(A^{\prime}(1^\lambda, \rho_0^{\otimes t})=0)+\frac{1}{2}\Pr(A^{\prime}(1^\lambda, \rho_1^{\otimes t})=1) \geq \frac{1}{2}+\frac{1}{8q(n)},
      \\ \Rightarrow~ & \frac{1}{2}\Pr(A^{\prime}(1^\lambda, \rho_0^{\otimes t})=0)+\frac{1}{2}\left(1-\Pr(A^{\prime}(1^\lambda, \rho_1^{\otimes t})=0)\right) \geq \frac{1}{2}+\frac{1}{8q(n)},
      \\ \Rightarrow~  &\Pr(A^{\prime}(1^\lambda, \rho_0^{\otimes t})=0)-\Pr(A^{\prime}(1^\lambda, \rho_1^{\otimes t})=0) \geq \frac{1}{4q(n)}.
    \end{align*} 
This along with Fact \ref{fact:indistinguishability_multiple_samples} contradicts that $\rho_0\approx_{\negl_C}\rho_1$. 
\end{proof}

\section*{Acknowledgment}
\noindent R.B. would like to thank Naresh Goud Boddu, Upendra Kapshikar and Sayantan Chakraborty for useful discussions. R.J. would like to thank Fang Song for pointing out the reference~\cite{cavalar2023computational}. We would also like to thank Marco Tomamichel for pointing out that a result in the reference~\cite{Renyi_chain_rules} was incorrectly stated.   The work of R.J. is supported by the NRF grant NRF2021-QEP2-02-P05 and
the Ministry of Education, Singapore, under the Research Centres of Excellence
program. This work was done in part while R.J. was visiting the Simons
Institute for the Theory of Computing, Berkeley, CA, USA.  R.B. is supported by the Singapore Ministry of Education and the
National Research Foundation through the core grants of the Centre for Quantum Technologies.


\bibliography{name}
\bibliographystyle{alpha}

\appendix
\section{Proof of Claim \ref{claim:combined}}
\begin{proof}
\noindent Let $i^*$ be from Claim~\ref{claim:D_close}. Let $l_1=\mathrm{D}(\tau^{XQ^{i^*+1}}\|\Id^X\otimes\tau^{Q^{i^*+1}}),~l_2=\mathrm{D}(\tau^{XQ^{i^*}}\|\Id^X\otimes\tau^{Q^{i^*}})$. Define \begin{equation}
    \label{eq:define_theta}
\theta^{XQ^m}\defeq(\tau^{Q^{i^*}})^{\frac{-1}{2}}\tau^{XQ^m}(\tau^{Q^{i^*}})^{\frac{-1}{2}}.
\end{equation} Note that this is not necessarily a state but it is a PSD matrix. Consider its spectral decomposition (on $XQ^{i^*}$) of the form: 
\begin{equation*}
XQ^{i^*}_{\theta}=\sum_{x,k}p_{x,k}\ketbra{x}{x}^{X}\otimes \ketbra{e_{x,k}}{e_{x,k}}^{Q^{i^*}}.
\end{equation*}
From the above expression, we can extend from $XQ^{i^*}$ to $XQ^m$ as below: \begin{equation}
XQ^m_{\theta}=\sum_{x,k}p_{x,k}\ketbra{x}{x}^{X}\otimes \ketbra{e_{x,k}}{e_{x,k}}^{Q^{i^*}} \otimes \tau_x^{Q^{i^*+1}\cdots Q^m}. \label{eq:natural_extension}
\end{equation}
Consider a function $J$ where $J(p_{x,k})=r$ if $p_{x,k}\in \Big(\frac{1}{2^{r}},\frac{1}{2^{r-1}}\Big]$ for $r\in[n^{2+c''}]$ and $J(p_{x,k})=0$ otherwise. Here $|J|\leq (c''+2)\log(n)$. We have
\begin{align*}
    {\theta_j}^{XQ^m}&=\sum_{x,k :J(p_{x,k})=j}p_{x,k}\ketbra{x}{x}^{X}\otimes \ketbra{e_{x,k}}{e_{x,k}}^{Q^{i^*}}\otimes \tau_x^{Q^{i^*+1}\cdots Q^m}.
\end{align*}
Then \[
\theta^{XQ^m}=\sum_j\theta_j^{XQ^m}.
\] We see that {\begin{equation}
\Tr(\theta^{XQ^m}_{0})=\Pr(J=0)_\theta\leq \frac{2^{n^{c''}}}{2^{n^{c''+2}}} =\negl(n). \label{eq:prob_zero_negl}
\end{equation}}
From eq. \ref{eq:define_theta}, using $\supp(\tau^{XQ^m})\subseteq \supp(\tau^{Q^{i^*}})$, we get 
\begin{equation*}
{\tau}^{XQ^{m}}=(\tau^{Q^{i^*}})^{\frac{1}{2}}{\theta}^{XQ^{m}}(\tau^{Q^{i^*}})^{\frac{1}{2}}=\sum_j(\tau^{Q^{i^*}})^{\frac{1}{2}}\theta_j^{XQ^m}(\tau^{Q^{i^*}})^{\frac{1}{2}}.
\end{equation*}
We can extend ${\tau}$ as \begin{equation}
{\tau}^{JXQ^{m}}\defeq\sum_j p_j\ketbra{j}{j}^J\otimes{\tau}_j^{XQ^{m}}, \label{eq:tau_extension}
\end{equation}
where \[p_j\defeq\Tr((\tau^{Q^{i^*}})^{\frac{1}{2}}\theta_j^{XQ^m}(\tau^{Q^{i^*}})^{\frac{1}{2}}), \]\[\sum_jp_j=1,\]\[\tau_j^{XQ^{m}}\defeq\frac{(\tau^{Q^{i^*}})^{\frac{1}{2}}\theta_j^{XQ^m}(\tau^{Q^{i^*}})^{\frac{1}{2}}}{\Tr((\tau^{Q^{i^*}})^{\frac{1}{2}}\theta_j^{XQ^m}(\tau^{Q^{i^*}})^{\frac{1}{2}})}.\] 
Now {\begin{align}   p_{0}&=\Tr((\tau^{Q^{i^*}})^{\frac{1}{2}}\theta_{0}^{XQ^m}(\tau^{Q^{i^*}})^{\frac{1}{2}}) \nonumber \\&=\Tr(\theta_{0}^{XQ^m}\tau^{Q^{i^*}})&\mbox{(cyclicity)} \nonumber
\\&\leq \Tr(\theta_{0}^{XQ^m}\Id^{Q^{i^*}}) \nonumber
\\&=\negl(n). &\mbox{(eq. \ref{eq:prob_zero_negl})} \label{eq:p_j=0_negl}
\end{align}}Note that for the inequality above, we have used that for operators $A,B, M$ such that $A\geq B$ and $M\geq 0$,  $\Tr (M A) \geq \Tr (MB)$.

Consider{\begin{align}
&\support((\tau^{Q^{i^*}})^{\frac{1}{2}}\Pi_{\support({\tau}_j^{XQ^{i^*}})} (\tau^{Q^{i^*}})^{\frac{1}{2}})\nonumber\\&=\support((\tau^{Q^{i^*}})^{\frac{-1}{2}}\Pi_{\support({\tau}_j^{XQ^{i^*}})} (\tau^{Q^{i^*}})^{\frac{-1}{2}})\nonumber &\mbox{(support of the pseudo-inverse)}\nonumber
\\&=\support((\tau^{Q^{i^*}})^{\frac{-1}{2}}({\tau}_j^{XQ^{i^*}})(\tau^{Q^{i^*}})^{\frac{-1}{2}})\nonumber &\mbox{(since $\tau^{Q^{i^*}}\geq 0$ and ${\tau}_j^{XQ^{i^*}}\geq 0$)}\nonumber
\\&=  {\support(\Pi_{\support({\tau}^{Q^{i^*}})}\theta_{j}^{XQ^{i^*}}\Pi_{\support({\tau}^{Q^{i^*}})})} \nonumber
\\&={\support(\theta_{j}^{XQ^{i^*}})}. \label{eq:supports}
\end{align}}
For all $j\neq 0$, there exists some $\alpha>0$ such that \begin{equation*}
\alpha\Pi_{\support(\theta_{j}^{XQ^{i^*}})} \leq  \theta_{j}^{XQ^{i^*}} \leq  2\alpha\Pi_{\support(\theta_{j}^{XQ^{i^*}})} \leq 2\alpha \Id^{XQ^{i^*}} .
\end{equation*}
From Fact \ref{fact:conjugation}, we have \begin{align}
&\alpha\Pi_{\support(\tau^{Q^{i^*}})}\Pi_{\support(\theta_{j}^{XQ^{i^*}})}\Pi_{\support(\tau^{Q^{i^*}})} \leq  (\tau^{Q^{i^*}})^{\frac{-1}{2}}{\tau}_{j}^{XQ^{i^*}}(\tau^{Q^{i^*}})^{\frac{-1}{2}} 
 \nonumber\\&\tab\tab\tab\tab\leq  2\alpha\Pi_{\support(\tau^{Q^{i^*}})}\leq 2\alpha \Id^{XQ^{i^*}}.\label{eq:projections}
\end{align}
Using Fact \ref{fact:conjugation} and  $\supp(\tau^{XQ^m})\subseteq \supp(\tau^{Q^{i^*}})$, we get
\[
\alpha(\tau^{Q^{i^*}})^{\frac{1}{2}}\Pi_{\support(\theta_{j}^{XQ^{i^*}})} (\tau^{Q^{i^*}})^{\frac{1}{2}}\leq  {\tau}_{j}^{XQ^{i^*}}.
\]
Using Fact \ref{fact:conjugation}, we have
\[
\alpha\Pi_{\support({\tau}_{j}^{XQ^{i^*}})} (\tau^{Q^{i^*}})^{\frac{1}{2}}\Pi_{\support(\theta_{j}^{XQ^{i^*}})} (\tau^{Q^{i^*}})^{\frac{1}{2}} \Pi_{\support({\tau}_{j}^{XQ^{i^*}})} \leq  {\tau}_{j}^{XQ^{i^*}}.
\]
Tracing out on both sides and taking $\log$ gives \begin{align*}
0&\geq\log \alpha+ \log \Tr\left(\Pi_{\support({\tau}_{j}^{XQ^{i^*}})} (\tau^{Q^{i^*}})^{\frac{1}{2}}\Pi_{\support(\theta_{j}^{XQ^{i^*}})} (\tau^{Q^{i^*}})^{\frac{1}{2}} \Pi_{\support({\tau}_{j}^{XQ^{i^*}})} \right)\\ &= \log \alpha+ \log \Tr \left(\Pi_{\support(\theta_{j}^{XQ^{i^*}})} (\tau^{Q^{i^*}})^{\frac{1}{2}} \Pi_{\support({\tau}_{j}^{XQ^{i^*}})} (\tau^{Q^{i^*}})^{\frac{1}{2}}\right) &\mbox{(cyclicity of trace)}
\\ &= \log \alpha+ \log \Tr \left( (\tau^{Q^{i^*}})^{\frac{1}{2}} \Pi_{\support({\tau}_{j}^{XQ^{i^*}})} (\tau^{Q^{i^*}})^{\frac{1}{2}}\right) &\mbox{(eq. \ref{eq:supports})}
\\ &= \log \alpha+ \log \Tr \left( \Pi_{\support({\tau}_{j}^{XQ^{i^*}})} \tau^{Q^{i^*}}\right).
\end{align*}
This gives \begin{align}
-\log\Tr\left(\Pi_{\support({\tau}_j^{XQ^{i^*}})} \tau^{Q^{i^*}}\right) & \geq \log(\alpha) .\label{eq:D_ineq_1}
\end{align}
We also have \begin{align} \dtwo{{\tau}_j^{XQ^{i^*}}}{\Id^X\otimes\tau^{Q^{i^*}}} &= \log\Tr\left((\tau^{Q^{i^*}})^{\frac{-1}{2}}{\tau}_j^{XQ^{i^*}}(\tau^{Q^{i^*}})^{\frac{-1}{2}} {\tau}_j^{XQ^{i^*}}\right) &\mbox{(Definition \ref{def:renyi_conditional})}\nonumber
\\&\leq \log\Tr(2 \alpha\Id^{XQ^{i^*}} {\tau}_j^{XQ^{i^*}})&\mbox{(eq. \ref{eq:projections})}\nonumber
    \\&\leq \log(\alpha)+1. \label{eq:D_ineq_11}
\end{align}
Combining eqs. \ref{eq:D_ineq_1} and \ref{eq:D_ineq_11}, for all $j\neq0$, we get, 
\begin{equation}
    \label{eq:closeness}
\dtwo{{\tau}_j^{XQ^{i^*}}}{\Id^X\otimes\tau^{Q^{i^*}}}+\log\Tr\left(\Pi_{\support({\tau}_j^{XQ^{i^*}})} \tau^{Q^{i^*}}\right) \leq 1.
\end{equation}
Define \[l_{2,j} \defeq\diver{{\tau}_j^{XQ^{i^*}}}{\Id^X\otimes\tau^{Q^{i^*}}}.\]
Consider the eigendecomposition 
\[(\tau^{Q^{i^*+1}})^{\frac{-1}{2}}{\tau}_j^{XQ^{i^*+1}}(\tau^{Q^{i^*+1}})^{\frac{-1}{2}}=\sum_i\lambda_{i,j}\ketbra{e_{i,j}}{e_{i,j}}^{XQ^{i^*+1}}.\] 
Define \[\good_j \defeq\{i: \lambda_{i,j}<2^{t_{j}}, ~t_{j}=l_{2,j}+(3c^{-1}+10c'')\log(n)\},\]
 \begin{equation} 
\Theta_j \defeq\sum_{i\in \good_j}\lambda_{i,j}\ketbra{e_{i,j}}{e_{i,j}}^{XQ^{i^*+1}}, \label{eq:def_big_theta}
\end{equation}
\[
\bar{\Theta}_j \defeq \sum_{i\notin \good_j}\lambda_{i,j}\ketbra{e_{i,j}}{e_{i,j}}^{XQ^{i^*+1}}.
\]
Define (the normalized state) \begin{equation}  \sigma_j^{XQ^{i^*+1}}\defeq\frac{(\tau^{Q^{i^*+1}})^{\frac{1}{2}}\Theta_j^{XQ^{i^*+1}}(\tau^{Q^{i^*+1}})^{\frac{1}{2}}}{\Tr((\tau^{Q^{i^*+1}})^{\frac{1}{2}}\Theta_j^{XQ^{i^*+1}}(\tau^{Q^{i^*+1}})^{\frac{1}{2}})}. \label{eq:def_sigma_j}
\end{equation}
{We extend $\sigma_j^{{XQ^{i^*+1}}}$ to $\sigma_j^{{XQ^{m}}}$ in the same way as eq.~\ref{eq:natural_extension}.}

\noindent Consider eq. \ref{eq:tau_extension}. Let us assume for contradiction that for all $p_j>\frac{1}{n^{10+c''}}$,
\[weight(\good_j)\defeq\Tr((\tau^{Q^{i^*+1}})^{\frac{1}{2}}\Theta_j^{XQ^{i^*+1}}(\tau^{Q^{i^*+1}})^{\frac{1}{2}})< \frac{1}{n^{10}}.\]

\noindent Define \[B\defeq\{j:p_j\leq \frac{1}{n^{10+c''}}\}.\]
Using $|J|\leq (c''+2)\log(n)$, we have\[\sum_{j\in {B}}p_j\leq \frac{2^{|J|}}{n^{10+c''}}\leq \frac{1}{n^8}.\] This gives \begin{equation}
    \sum_{j\in \widebar{B}}p_j\geq 1-\frac{1}{n^8}.  \label{eq:prob_bar_b}
\end{equation}
Define {unnormalised} $\rho_j^{XQ^{i^*+1}}\defeq(\tau^{Q^{i^*+1}})^{\frac{1}{2}}\bar{\Theta}_j^{XQ^{i^*+1}}(\tau^{Q^{i^*+1}})^{\frac{1}{2}}$ where \begin{equation}
\bar{\Theta}_j^{XQ^{i^*+1}} \geq 2^{t_{j}} \Pi_{\support(\bar{\Theta}_j^{XQ^{i^*+1}})}. \label{eq:theta-bar}
\end{equation}
Let \[
\rho_{0,j}^{XQ^{i^*+1}}\defeq\frac{\rho_j^{XQ^{i^*+1}}}{\Tr(\rho_j^{XQ^{i^*+1}})}.
\]
For all $j\in \Bar{B}$, we have\[
{\tau}_j^{XQ^{i^*+1}}=\alpha\rho_{0,j}^{XQ^{i^*+1}}+(1-\alpha)\sigma_j^{XQ^{i^*+1}},
\]
where \begin{align}
    \alpha=\Tr(\rho_j^{XQ^{i^*+1}})\geq 1-\frac{1}{n^{10}}. \label{eq:alpha_weight}
\end{align}
Using Fact \ref{fact:conjugation} and eq. \ref{eq:theta-bar}, we get\[
(\tau^{Q^{i^*+1}})^{\frac{-1}{2}}\rho_j^{XQ^{i^*+1}}(\tau^{Q^{i^*+1}})^{\frac{-1}{2}} \geq 2^{t_{j}} \Pi_{\support(\tau^{Q^{i^*+1}})}\Pi_{\support(\bar{\Theta}_j^{XQ^{i^*+1}})}\Pi_{\support(\tau^{Q^{i^*+1}})}
\] which also gives \begin{equation}
    \label{eq:as_before}
(\tau^{Q^{i^*+1}})^{\frac{-1}{2}}\rho_{0,j}^{XQ^{i^*+1}}(\tau^{Q^{i^*+1}})^{\frac{-1}{2}} \geq 2^{t_{j}} \Pi_{\support(\tau^{Q^{i^*+1}})}\Pi_{\support(\bar{\Theta}_j^{XQ^{i^*+1}})}\Pi_{\support(\tau^{Q^{i^*+1}})}.
\end{equation}
As before, we have
\begin{align}
&\support((\tau^{Q^{i^*+1}})^{\frac{1}{2}}\Pi_{\support(\rho_{0,j}^{XQ^{i^*+1}})} (\tau^{Q^{i^*+1}})^{\frac{1}{2}})\nonumber\\&=\support((\tau^{Q^{i^*+1}})^{\frac{-1}{2}}\Pi_{\support(\rho_{0,j}^{XQ^{i^*+1}})} (\tau^{Q^{i^*+1}})^{\frac{-1}{2}})\nonumber
\\&=\support((\tau^{Q^{i^*+1}})^{\frac{-1}{2}}(\rho_{0,j}^{XQ^{i^*+1}})(\tau^{Q^{i^*+1}})^{\frac{-1}{2}})\nonumber
\\&=\support((\tau^{Q^{i^*+1}})^{\frac{-1}{2}}(\rho_j^{XQ^{i^*+1}})(\tau^{Q^{i^*+1}})^{\frac{-1}{2}}) \nonumber
\\&={\support(\bar{\Theta}_{j}^{XQ^{i^*+1}})}. \label{eq:supp_equal}
\end{align}
From eq. \ref{eq:as_before} and Fact \ref{fact:conjugation}, using $\support(\rho_{0,j}^{XQ^{i^*+1}})\subseteq \support(\tau^{Q^{i^*+1}})$, we have \[
\rho_{0,j}^{XQ^{i^*+1}}\geq 2^{t_{j}} (\tau^{Q^{i^*+1}})^{\frac{1}{2}}\Pi_{\support(\bar{\Theta}_j^{XQ^{i^*+1}})}(\tau^{Q^{i^*+1}})^{\frac{1}{2}}.
\]
Fact \ref{fact:conjugation} gives us \[
\rho_{0,j}^{XQ^{i^*+1}}\geq 2^{t_{j}} \Pi_{\support(\rho_{0,j}^{XQ^{i^*+1}})}(\tau^{Q^{i^*+1}})^{\frac{1}{2}}\Pi_{\support(\bar{\Theta}_j^{XQ^{i^*+1}})}(\tau^{Q^{i^*+1}})^{\frac{1}{2}}\Pi_{\support(\rho_{0,j}^{XQ^{i^*+1}})}. 
\]
As before, tracing out on both sides and taking $\log$ gives \begin{align*}
0&\geq t_{j}+ \log \Tr\left(\Pi_{\support(\rho_{0,j}^{XQ^{i^*+1}})}(\tau^{Q^{i^*+1}})^{\frac{1}{2}}\Pi_{\support(\Bar{\Theta}_{j}^{XQ^{i^*+1}})} (\tau^{Q^{i^*+1}})^{\frac{1}{2}} \Pi_{\support(\rho_{0,j}^{XQ^{i^*+1}})} \right)\\ &= t_{j}+ \log \Tr \left(\Pi_{\support(\Bar{\Theta}_{j}^{XQ^{i^*+1}})} (\tau^{Q^{i^*+1}})^{\frac{1}{2}} \Pi_{\support(\rho_{0,j}^{XQ^{i^*+1}})} (\tau^{Q^{i^*+1}})^{\frac{1}{2}}\right) &\mbox{(cyclicity)}
\\ &= t_{j}+ \log \Tr \left( (\tau^{Q^{i^*+1}})^{\frac{1}{2}} \Pi_{\support(\rho_{0,j}^{XQ^{i^*+1}})} (\tau^{Q^{i^*+1}})^{\frac{1}{2}}\right) &\mbox{(eq. \ref{eq:supp_equal})}
\\ &= t_{j}+ \log \Tr \left( \Pi_{\support(\rho_{0,j}^{XQ^{i^*+1}})} \tau^{Q^{i^*+1}}\right). 
\end{align*}
This gives\begin{align}
    \diver{\rho_{0,j}^{XQ^{i^*+1}}}{\Id^X\otimes\tau^{Q^{i^*+1}}} & \geq -\log\Tr(\Pi_{\support(\rho_{0,j}^{XQ^{i^*+1}})}\tau^{Q^{i^*+1}}) \nonumber&\mbox{(Fact \ref{fact:D_0})}
    \\& \geq t_{j}. \label{eq:lower_bound_D}
\end{align}
We also have
\begin{align}
&\diver{\sigma_j^{XQ^{i^*+1}}}{U_X\otimes\tau^{Q^{i^*+1}}} \geq 0. &\mbox{(Fact \ref{fact:non_negative})} \nonumber\\
\implies& \diver{\sigma_j^{XQ^{i^*+1}}}{\Id^X\otimes\tau^{Q^{i^*+1}}} \geq -n. &\mbox{($|X|=n$)} \label{eq:lower-minus-n}\end{align}
Thus, for all $j \in \Bar{B}$, from Fact \ref{fact:quantum_chaining}, 
\begin{align}  \diver{{\tau}_j^{XQ^{i^*+1}}}{\Id_X\otimes\tau^{Q^{i^*+1}}}&\geq \alpha \diver{\rho_{0,j}^{XQ^{i^*+1}}}{\Id_X\otimes\tau^{Q^{i^*+1}}} +(1-\alpha) \diver{\sigma_j^{XQ^{i^*+1}}}{\Id_X\otimes\tau^{Q^{i^*+1}}} - 1 \nonumber
\\ &\geq (1-\frac{1}{n^{10}})t_{j} -\frac{n}{n^{10}} - 1 \nonumber
\\ &\geq (1-\frac{1}{n^{10}})t_{j} -2. \label{eq:lowerbounding}
\end{align}
The first inequality above follows from eqs. \ref{eq:alpha_weight}, \ref{eq:lower_bound_D} and eq. \ref{eq:lower-minus-n}.
From Fact \ref{fact:quantum_chaining}, we get \begin{align}
&\E_j\left[\mathrm{D}({\tau}_j^{XQ^{i^*}}\|\Id^X\otimes\tau^{Q^{i^*}})\right] \geq \mathrm{D}(\tau^{XQ^{i^*}}\|\Id^X\otimes\tau^{Q^{i^*}}). \nonumber\end{align}
We can rewrite this as
\begin{align}
     \sum_{j}p_jl_{2,j} \geq l_2, \label{eq:convexity}
\end{align}where $l_2=\mathrm{D}(\tau^{XQ^{i^*}}\|\Id^X\otimes\tau^{Q^{i^*}})$. {Using Fact \ref{fact:identity_upper} and eq. \ref{eq:tau_extension}, for $j\in \bar{B}$, we have \begin{equation}
\Id^X\otimes\tau^{Q^{i^*}}\geq \tau^{XQ^{i^*}}\geq \frac{1}{n^{10+c''}} {\tau_j}^{XQ^{i^*}}.
\label{eq:dmax_dominance}\end{equation}
Thus, for $j\in \bar{B}$, we get \begin{align}
    l_{2,j}&=\mathrm{D}({\tau}_j^{XQ^{i^*}}\|\Id^X\otimes\tau^{Q^{i^*}}) \nonumber\\
    &\leq \dmax{{\tau}_j^{XQ^{i^*}}}{\Id^X\otimes\tau^{Q^{i^*}}} &\mbox{(Fact \ref{fact:monotonicity_renyi})} \nonumber\\
    &\leq (10+c'')\log(n). &\mbox{(eq. \ref{eq:dmax_dominance})} \label{eq:l_2}
\end{align}
}We have\begin{align*}
    &\diver{\tau^{XQ^{i^*+1}}}{\Id^X \otimes \tau^{Q^{i^*+1}}}\\&\geq \E_j \left[\diver{{\tau}_j^{XQ^{i^*+1}}}{\Id^X \otimes \tau^{Q^{i^*+1}}}\right] - |J|&\mbox{(Fact \ref{fact:quantum_chaining})}
    \\&= \sum_{j\in \Bar{B}} p_j\diver{{\tau}_j^{XQ^{i^*+1}}}{\Id^X \otimes \tau^{Q^{i^*+1}}} +\sum_{j\in {B}} p_j\diver{{\tau}_j^{XQ^{i^*+1}}}{\Id^X \otimes \tau^{Q^{i^*+1}}} - |J|\\
    &\geq \sum_{j\in \bar{B}}p_j\left[\diver{{\tau}_j^{XQ^{i^*+1}}}{\Id^X \otimes \tau^{Q^{i^*+1}}} \right] + \sum_{j\in {B}}p_j\left[\mathrm{D}({\tau}_j^{XQ^{i^*}}\|\Id^X \otimes \tau^{Q^{i^*}}) \right] - |J| &\mbox{(Fact \ref{fact:data})}
    \\ & \geq \sum_{j\in \bar{B}}p_j\left[(1-\frac{1}{n^{10}})t_{j} -2  \right] +\sum_{j\in {B}}p_jl_{2,j}- |J| &\mbox{(eq. \ref{eq:lowerbounding})}
     \\ & \geq \sum_{j\in \bar{B}}p_j\left[(1-\frac{1}{n^{10}})t_{j} \right]-2   +\sum_{j\in {B}}p_jl_{2,j}- |J| &\mbox{($\sum_{j\in \bar{B}}p_j\leq 1$)}
     \end{align*}
     \begin{align*}
     & = \sum_{j\in \bar{B}}p_j\left[(1-\frac{1}{n^{10}})(l_{2,j}+(3c^{-1}+10c'')\log(n)) \right] -2  +\sum_{j\in {B}}p_jl_{2,j}- |J|  &\mbox{($t_{j}=l_{2,j}+(3c^{-1}+10c'')\log(n)$)} 
    \\ & \geq \sum_{j}p_jl_{2,j}+\left(1-\frac{1}{n^{10}}\right)\left(1-\frac{1}{n^{8}}\right)\cdot(3c^{-1}+10c'')\log(n) \\&\tab -\sum_{j\in \bar{B}}p_j\cdot\frac{(10+c'')\log(n)}{n^{10}}-2  - |J| &\mbox{(eqs. \ref{eq:prob_bar_b} and \ref{eq:l_2})}
   \\ & \geq \sum_{j}p_jl_{2,j}+\left(1-\frac{1}{n^{10}}\right)\left(1-\frac{1}{n^{8}}\right)\cdot(3c^{-1}+10c'')\log(n) -3 - |J| 
   \\&\geq l_2+\left(1-\frac{1}{n^{10}}\right)\left(1-\frac{1}{n^{8}}\right)\cdot(3c^{-1}+10c'')\log(n) -3 - |J| &\mbox{(eq. \ref{eq:convexity})}\\ 
   &\geq l_2 +2c^{-1} \log(n) &\mbox{($|J|\leq (c''+2)\log(n)$)}\\&= \mathrm{D}(\tau^{XQ^{i^*}}\|\Id^X\otimes\tau^{Q^{i^*}}) +2c^{-1} \log(n).
\end{align*}
This is a contradiction to Claim \ref{claim:D_close}. Thus, there exists a  $j$  such that $p_j>\frac{1}{n^{10+c''}}$ and \begin{align}
    weight(\good_j)=\Tr((\tau^{Q^{i^*+1}})^{\frac{1}{2}}\Theta_j^{XQ^{i^*+1}}(\tau^{Q^{i^*+1}})^{\frac{1}{2}})\geq \frac{1}{n^{10}}. \label{eq:weight_good_large}\end{align} Furthermore, eq. \ref{eq:p_j=0_negl} ensures that the $j$ we identified above is non-zero. From eq. \ref{eq:def_big_theta}, we have \[\Theta_j^{XQ^{i^*+1}}\leq 2^{t_{j}}\Id^{XQ^{i^*+1}}.\] 
Fact \ref{fact:conjugation} gives
 \[(\tau^{Q^{i^*+1}})^{\frac{1}{2}}\Theta_j^{XQ^{i^*+1}}(\tau^{Q^{i^*+1}})^{\frac{1}{2}}\leq 2^{t_{j}}\Id^{X} \otimes\tau^{Q^{i^*+1}}.\]
 Thus,
 \begin{align*}
 &\frac{(\tau^{Q^{i^*+1}})^{\frac{1}{2}}\Theta_j^{XQ^{i^*+1}}(\tau^{Q^{i^*+1}})^{\frac{1}{2}}}{\Tr((\tau^{Q^{i^*+1}})^{\frac{1}{2}}\Theta_j^{XQ^{i^*+1}}(\tau^{Q^{i^*+1}})^{\frac{1}{2}})}\leq 2^{t_{j}}\cdot n^{10}\cdot\Id^{X} \otimes\tau^{Q^{i^*+1}}. &\mbox{(eq. \ref{eq:weight_good_large})}\end{align*}
Using eq. \ref{eq:def_sigma_j}, we get\begin{equation}
    \sigma_j^{XQ^{i^*+1}}\leq 2^{t_{j}}\cdot n^{10}\cdot\Id^{X} \otimes\tau^{Q^{i^*+1}}. \label{eq:dmax_equation}
\end{equation}
Consider,
\begin{align}
    S_2(X|Q^{i^*+1})_{\sigma_j} &\geq S_\infty(X|Q^{i^*+1})_{\sigma_j}  &\mbox{(Fact \ref{fact:monotonicity_renyi})}\nonumber\\ &=-\min_{\rho^{Q^{i^*+1}}}\mathrm{D}_\infty(\sigma_j^{XQ^{i^*+1}}\|\Id^X\otimes \rho^{Q^{i^*+1}} ) &\mbox{(Definition \ref{def:renyi_conditional})}\nonumber
    \\ &\geq -\mathrm{D}_\infty(\sigma_j^{XQ^{i^*+1}}\|\Id^X\otimes \tau^{Q^{i^*+1}} ) \nonumber
    \\&\geq -t_{j}-10\log
    (n)&\mbox{(eq. \ref{eq:dmax_equation})}\nonumber\\&= -l_{2,j}-(10+3c^{-1}+10c'') \log(n). 
 &\mbox{($t_{j}=l_{2,j}+(3c^{-1}+10c'')\log(n)$)}\label{eq_s2}
\end{align}
Since $p_j\geq \frac{1}{n^{10+c''}}$, from eqs.  \ref{eq:tau_extension} and   \ref{eq:weight_good_large}, we have
\begin{align}
\tau^{XQ^{i^*}}\geq \frac{1}{n^{10+c''}}{\tau}_j^{XQ^{i^*}} \geq \frac{1}{n^{20+c''}}\sigma_j^{XQ^{i^*}}. 
\label{eq:inequality_tau_tilde}
\end{align}
Thus, {on extending $\sigma_j^{XQ^{i^*}}$ to $\sigma_j^{XQ^m}$ (like eq. \ref{eq:natural_extension})}, we get \begin{equation}
\dmax{\sigma_j}{\tau}\leq (20+c'')\log(n). \label{eq:condition_1} \end{equation}
From eq. \ref{eq:inequality_tau_tilde}, we get \begin{align}
   & \tau^{Q^{i^*}} \geq \frac{1}{n^{20+c''}}\sigma_j^{Q^{i^*}} \label{eq:inequality_tau}
   \end{align}
   Tensoring with $\Id^X$ and taking $\log$, we get
   \begin{align}
       &\log(\Id^X\otimes \tau^{Q^{i^*}}) \geq \log\left(\frac{\Id^X\otimes\sigma_j^{Q^{i^*}}}{n^{20+c''}}\right)\nonumber
\\\implies &\Tr\left(\sigma_j^{XQ^{i^*}}\left(\log(\Id^X\otimes \tau^{Q^{i^*}}) -\log(\Id^X\otimes\sigma_j^{Q^{i^*}})\right)\right)\geq -(20+c'')\log(n). \label{eq:ineq_D2}
\end{align}
Note that above, we have used that for operators $A,B, M$ such that $A\geq B$ and $M\geq 0$,  $\Tr (M A) \geq \Tr (MB)$.
Now,
\begin{align}
\diver{\sigma_j^{XQ^{i^*}}}{\Id^X\otimes\tau^{Q^{i^*}}} & \geq -\log\Tr(\Pi_{\supp(\sigma_j^{XQ^{i^*}})}(\Id^X\otimes\tau^{Q^{i^*}})) &\mbox{(Fact \ref{fact:D_0})} \nonumber
    \\  & \geq -\log\Tr(\Pi_{\supp({\tau}_j^{XQ^{i^*}})}(\Id^X\otimes\tau^{Q^{i^*}})) &\mbox{(eq. \ref{eq:inequality_tau_tilde})}\nonumber
    \\  & \geq \dtwo{\tau_j^{XQ^{i^*}}}{(\Id^X\otimes\tau^{Q^{i^*}})} -1 &\mbox{(eq. \ref{eq:closeness})}\nonumber
     \\  & \geq \diver{\tau_j^{XQ^{i^*}}}{(\Id^X\otimes\tau^{Q^{i^*}})}  -1  &\mbox{(Fact \ref{fact:monotonicity_renyi})}\nonumber
    \\& =l_{2,j} -1. \label{eq:ineq_D}
\end{align}
Consider,
\begin{align}
&S(X|Q^{i^*})_{\sigma_j}\nonumber\\&=-\diver{\sigma_j^{XQ^{i^*}}}{\Id^X\otimes\sigma_j^{Q^{i^*}}} \nonumber
    \\&= -\Tr\left(\sigma_j^{XQ^{i^*}}\left(\log(\sigma_j^{XQ^{i^*}})-\log(\Id^X\otimes\sigma_j^{Q^{i^*}})\right)\right) \nonumber
    \\&= -\Tr\left(\sigma_j^{XQ^{i^*}}\left(\log(\sigma_j^{XQ^{i^*}})-\log(\Id^X\otimes\tau^{Q^{i^*}})\right)\right)\nonumber\\ &\tab- \Tr\left(\sigma_j^{XQ^{i^*}}\left(\log(\Id^X\otimes\tau^{Q^{i^*}})-\log(\Id^X\otimes\sigma_j^{Q^{i^*}})\right)\right) \nonumber
     \\&= -\diver{\sigma_j^{XQ^{i^*}}}{\Id^X\otimes\tau^{Q^{i^*}}}-\Tr\left(\sigma_j^{XQ^{i^*}}\left(\log(\Id^X\otimes\tau^{Q^{i^*}})-\log(\Id^X\otimes\sigma_j^{Q^{i^*}})\right)\right) \nonumber
     \\&\leq -l_{2,j} +(20+c'')\log(n)+1. &\mbox{(eqs. \ref{eq:ineq_D2} and \ref{eq:ineq_D})}\label{eq:s_upper1}
\end{align} 
From eqs. \ref{eq_s2} and \ref{eq:s_upper1}, we get
\begin{align}
 S(X|Q^{i^*})_{\sigma_j}- S_2(X|Q^{i^*+1})_{\sigma_j}\leq (30+3c^{-1}+11c'')\log(n)+1. \label{eq:condition_2}
\end{align}
From eqs. \ref{eq:condition_1} and \ref{eq:condition_2}, on setting $\gamma=\sigma_j$, we get         \[\dmax{\gamma}{\tau}\leq (20+c'')\log(n),\]  
       \[
S(X|Q^{i^*})_{\gamma}-S_2(X|Q^{i^*+1})_{\gamma} \leq (30+3c^{-1}+11c'')\log(n)+1. \qedhere\]
\end{proof}

\end{document}